\documentclass[
	draft=false,a4paper
	,11pt
	,DIV=8
	,parskip=half*
	,bibliography=totocnumbered
	]{scrartcl}

\pdfoutput=1		
 
\usepackage[utf8]{inputenc}
\usepackage[T1]{fontenc}
\usepackage[USenglish]{babel}

\usepackage{graphicx}
\usepackage[pdfusetitle]{hyperref}		
  
\usepackage{pdf14} 
\usepackage{amsmath}
\usepackage{amssymb}		
\usepackage{amsthm}

\usepackage{algorithm}			
\usepackage[]{algorithmic}

\usepackage{tikz}
\usepackage{tkz-graph}   
\usepackage{tkz-berge}

\usepackage[numbers]{natbib}
\usepackage{url,doi}

\usepackage{paralist}			
\usepackage[colorinlistoftodos]{todonotes}

\newcommand*{\sepsign}{\;{\footnotesize $\blacklozenge$}\; }
\newcommand*{\oneto}[1]{ \{1, \dots, #1\} }
\DeclareMathOperator{\dcup}{\dot{\cup}}
\newcommand*{\R}{\mathbf{R}}
\newcommand*{\Z}{\mathbf{Z}}
\newcommand*{\mS}{\mathcal{S}}
\newcommand*{\mC}{\mathcal{C}}

\newcommand*{\formatmathnames}[1]{\textnormal{\small #1}}
\newcommand*{\pfp}{\formatmathnames{PFP}} 
\newcommand*{\rpfp}{\formatmathnames{RPFP}} 
\newcommand*{\snpp}{\formatmathnames{SNPP}} 
\newcommand*{\rap}{\formatmathnames{RAP}} 
\newcommand*{\urap}{\formatmathnames{card-RAP}} 

\newcommand*{\threesat}{\formatmathnames{3-SAT}}
\newcommand*{\setcover}{\formatmathnames{SCP}}
\newcommand*{\vcthree}{\formatmathnames{VC\,3}}
\newcommand*{\opt}{\formatmathnames{OPT}}
\newcommand*{\val}{\formatmathnames{VAL}}
\newcommand*{\alg}{\formatmathnames{ALG}}
\newcommand*{\p}{\formatmathnames{P}}
\newcommand*{\np}{\formatmathnames{NP}}
\newcommand*{\dtime}{\formatmathnames{DTIME}}
\newcommand*{\lp}{\formatmathnames{LP}}
\newcommand*{\ilp}{\formatmathnames{ILP}}

\newcommand*{\uset}{\mathcal{F}} 
\newcommand*{\rset}{X} 
\newcommand*{\fset}{f} 
\newcommand*{\Fset}{F} 
\newcommand*{\inst}{\mathcal{I}} 

\theoremstyle{plain}
\newtheorem{theorem}{Theorem}
\newtheorem*{theorem*}{Theorem}

\newtheorem{lemma}[theorem]{Lemma}

\newtheorem{proposition}[theorem]{Proposition}

\theoremstyle{definition}
\newtheorem{definition}[theorem]{Definition}
\newtheorem{example}[theorem]{Example}

\newtheorem*{problem*}{Problem}

\theoremstyle{remark}
\newtheorem{remark}[theorem]{Remark}

\author{%
	\texorpdfstring		
	{
	  David Adjiashvili$^{\,\textrm{a}}$
	  \and
	  Viktor Bindewald$^{\,\textrm{b}}$
	  \and
	  Dennis Michaels$^{\,\textrm{b}}$}
	{David Adjiashvili, Viktor Bindewald, Dennis Michaels}
}

\publishers{\centering 
  \begin{minipage}{0.98\textwidth}
    \small
    \begin{itemize}
      \setlength{\itemsep}{0ex}
    \item[$^{\textrm{a}}$] Institute for Operations Research,
      Department of Mathematics, ETH 
      Zurich, Raemi\-strasse 101, 8092 Zurich (Switzerland) 
    \item[$^{\textrm{b}}$] Department of Mathematics, TU Dortmund University,
      Vogelpothsweg 87, 44227 Dortmund (Germany)
    \end{itemize}
  \end{minipage}
}

\title{\Large%
  Robust Assignments via Ear Decompositions and Randomized
  Rounding\footnote{An extended abstract of this manuscript appeared
    in I. Chatzigiannakis, M. Mitzenmacher, Y. Rabani, D. Sangiorgi
    (Eds.), Proceedings of the 43rd International Colloquium on
    Automata, Languages, and Programming (ICALP 2016), July  12-15, 2016, Rome (Italy).}
}

\date{}

\setkomafont{title}{\normalfont\tiny\bfseries}
\setkomafont{section}{\normalfont\large\bfseries}
\setkomafont{subsection}{\normalfont\normalsize\bfseries}
\setkomafont{paragraph}{\normalfont\normalsize\itshape}

\begin{document} 
   
\maketitle

\begin{abstract}
Many real-life planning problems require making a priori decisions before all parameters of the problem have been revealed. 
An important special case of such problem arises in scheduling problems, where a set of tasks 
needs to be assigned to the available set of machines or personnel (resources), in a way 
that all tasks have assigned resources, and no two tasks share the same resource. 
In its nominal form, the resulting computational problem becomes the 
\emph{assignment problem} on general bipartite graphs. 
  
This paper deals with a robust variant of  the assignment problem modeling 
situations where certain edges in the corresponding graph are 
\emph{vulnerable} and may become unavailable after a solution has been chosen. 
The goal is to choose a minimum-cost collection of edges such that if any 
vulnerable edge becomes unavailable, the remaining part of the solution contains 
an assignment of all tasks.  

We present approximation results and hardness proofs 
for this type of problems, and establish several connections to well-known 
concepts from matching theory, robust optimization and LP-based techniques.
\end{abstract}


\section{Introduction}
\label{sec:intro}

The need for incorporating system reliability into decision making has sprung wide-spread interest in optimization models which incorporate data uncertainty in the last decades. 
The latter trend has lead to the development of several new theories including the popular field of \emph{Robust Optimization}.
In robust optimization the nominal optimization problem is equipped with a set of \emph{scenarios}, representing various possible states of nature that may occur after the solution to the problem is chosen. 
The goal is to determine a solution that performs well, in terms of feasibility or cost, in a \emph{worst case} realization of the state of nature. 

The \emph{Assignment Problem} is one of the most fundamental optimization problems arising in many reliability-sensitive systems.  
In its nominal form, the input consists of a collection of $n_T$ \emph{tasks}, a set of $n_R$ \emph{resources}  (with $n_T\leq n_R$), and \emph{assignment costs} $c_{i,j}$ representing the cost associated with assigning resource $i$ to task $j$.  
The set of allowed assignments can be represented by a bipartite graph $G := (R\dcup T,E)$ where each resource $i$ corresponds to a node $r_i\in R$, each 
task $j$ corresponds to a node $t_j\in T$, and the edge $\{r_i, t_j\}$ is present 
in $E$ if the $j$-th task can be assigned to resource $i$. 
The goal is to find a minimum-cost \emph{matching} $M \subseteq E$ that covers all
nodes in $T$, i.e.\ a set of non-adjacent edges that is incident to every node in $T$.
In the following, a subset $M$ satisfying that property is called \emph{an assignment}.
A thorough introduction to the assignment problem can be 
found in the book of Burkard, Dell'{}Amico, and Martello~\cite{burkard_book_12}.

This paper deals with a natural robust counterpart of the assignment problem that is defined as follows. 
An instance of \rap\ consists of a bipartite graph $G=(R\dcup T,E)$ and a non-negative cost vector $c\in\R_{\geq 0}^E$ representing a nominal assignment problem. 
Furthermore, a collection $\uset\subseteq 2^E$ of subsets of edges are given where each $\Fset\in\mathcal{F}$ induces a failure scenario that leads to a deletion of $\Fset$ from $G$, i.e.\ if scenario $\Fset\in\uset$ emerges, all edges from $F$ are deleted from $G$.
The goal is to find a \emph{redundant assignment} $\rset\subseteq E$ of minimum cost 
with the property that, for every $\Fset\in\uset$, the set $\rset\setminus\Fset$ contains 
an assignment in $G$ (i.e.\ the graph $(R\dcup T, X\setminus \Fset)$ contains an assignment).  

The robustness paradigm considered in this paper fits into the concept
of \emph{redundancy-based robustness} --  a well-motivated and widely
studied approach 
(See Bertsimas, Brown and Caramanis~\cite{BertsimasBrownCaramanis}, and Herroelen and Leus~\cite{herroelen2005project} for an overview of different robustness concepts).
Some of the problems falling into this category include the minimum
$k$-edge connected spanning subgraph problem (Cheriyan, Seb{\H{o}} and Szigeti~\cite{cheriyan_et_al_01}, Gabow et al.~\cite{GabowGoemansTardosWilliamson}) and the robust facility location
problem (Jain and Vazirany~\cite{jain2000approximation}, Swamy and Shmoys \cite{swamy2003fault}, Chechik and Peleg~\cite{chechik2010robust}).  

Adjiashvili, Stiller and Zenklusen~\cite{AdjishviliStillerZenklusen2014} introduce a robustness model called \emph{bulk-robustness}, which combines the standard redundancy based robustness approach with a non-uniform failure model. 
In its general form, a bulk-robust counterpart of a combinatorial optimization problem consists of an instance of the nominal problem, as well as a collection of scenarios, each comprising an arbitrary set of resources that may fail simultaneously. 
The goal is to determine a minimum-cost set of resources that contains a feasible solution, even when the resources in any single failure scenario become unavailable. In the language of bulk-robustness, \rap\ is the bulk-robust assignment 
problem.

It is important to remark that several other robust counterparts of
the assignment problem have been considered in the literature under the same, or similar, names. A
brief review on relevant models and works existing in the literature
is given in Section~\ref{sec:related-work}.

In the remainder of this section, a few motivating applications for \rap\ are provided, some connections to related notions in matching theory are established, and main results as well as technical contributions of this paper are briefly discussed.

\subsection{Motivation}\label{subsec:motivation}

The most natural applications of \rap, and redundancy-based robust optimization
in general, emerge in situations where resources can not be 
easily made available on demand. 
In such applications, any resource intended for deployment at 
a certain point in time must be \emph{reserved} at an earlier 
stage, and thus made available for potential deployment. 
Examples of such applications range from the construction of robust 
power transmission networks (Hajiaghayi, Immorlica, and 
Mirrokni~\cite{hajiaghayi2003power}) to supply chain management (Tang~\cite{tang2006robust}).

While we think that \rap\ can be a useful model to incorporate 
robustness in any assignment model with up-front decisions of the 
latter type,  a few, more concrete, applications are brought
hereafter.

\paragraph{Flexible Designs for Manufacturing Processes.}
Flexible designs for manufacturing processes have attracted 
significant attention in the operations research community in recent years. 
The motivation in this topic is the need of manufacturing systems 
in modern economies to be able to adapt to quickly changing demand patterns.
The flexible design is modeled as a problem of selecting a set of edges 
in a bipartite graph with sides corresponding to plants and products, respectively.
An edge between a plant and a product means that the plant can 
produce the corresponding product. While plants have fixed deterministic capacities, 
the demand for products can vary, and is assumed to either be random, or 
materialize from a known uncertainty set. The goal is to choose a cheap set 
of edges that will maximize the expected, or the worst-case flow between 
the two sides of the graph. In this view, \rap\ can be seen as a problem in 
process flexibility, in which uncertainty lies in the structure of the graph,
instead of the demand patters, and the goal is to always satisfy the full demand.
For an overview of related results we refer to the
papers of Chou et al.~\cite{chou2010design}, Simchi-Levi and
Wei~\cite{simchi2012understanding,simchi2015worst}, and 
D{\'e}sir et al.~\cite{desir2016sparse}.

\paragraph{Staff Training.}
Large companies often employ intensive training programs for their 
employees, designed to adapt the available pool of skills to their dynamic needs. 
For instance, developing new software products often requires 
improved knowledge on recent technologies that employees have to 
be trained for.  
It is natural to incorporate the incurred training costs into the task 
allocation problem which, in turn, naturally corresponds to an assignment problem. 
The cost of assigning an employee to perform a given task in the 
project corresponds to the training cost incurred by training the 
employee to perform this task. \\
In a more realistic scenario, some employee to task assignments might 
become unavailable \emph{even if the employee were trained to perform the task}. 
That type of vulnerability is very common, and can be caused, e.g.\ by 
employee dissatisfaction from his task assignment and by unexpected 
inability (due to injury or unavailability of equipment, etc.). 
\rap\, is a suitable model for deciding on robust training programs 
for the project, where skill sets of the employees allow for reassignments 
even if some employee to task assignment becomes unavailable.
  
\paragraph{Continuity of Service.}
In industries such as health care and consulting, it is often 
desirable to maintain very stable client to operator relationships.
A typical example is a nursing home where elderly people feel more 
secure and comfortable if the nurses taking care of them do not change often.  
This is called \emph{continuity of care} in the health care literature 
(Carello and Lanzarone~\cite{carello_lanzarone_14}).
In the context of consulting the benefit is the reduction of the loss
of undocumented knowledge and employees' stress levels due to
decreased number of reassignments. Indirectly this also improves the
customer's experience.\\
Again, the underlying resource allocation problem can be modeled with the help of a bipartite 
graph $G=(R \dcup T, E)$ where nodes from $R$ represent the staff 
members, while a node from $T$ corresponds to a certain \emph{task}. 
In this setting, a task can now be any combination of a patient, 
a shift, and a certain type of service.
The edges in $G$ encode whether a staff member is able to perform a 
task, and the uncertainty set $\uset \subseteq 2^E$ contains all 
edges for which the corresponding staff member has to change a shift 
with high probability for some reason.
From the continuity of service perspective, the aim is to find a 
smallest subset $\rset$ of edges in $G$ containing an assignment 
not using $\Fset$, for every $\Fset\in\uset$.
This problem gives rise to an unweighted \rap\ instance on $G$ and scenario set $\uset$.

\subsection{Overview of results and techniques}

This paper addresses the computational complexity of \rap. 
The main contributions are approximation algorithms and
hardness of approximation results. 
The study of approximation algorithms is especially 
justified by showing that \rap\ is \np-hard, even in very restricted variants.

\paragraph{Notation.}
To keep notation short, the following standard notation from graph theory is frequently used throughout the paper.
Let a graph $G:=(R\dcup T, E)$ and subsets $H\subseteq E$ and $U\subseteq R\dcup T$ be given. 
The set of all nodes in the graph $G$ covered by edge set $H$ is abbreviated by $V[H]$, and the set of all edges in $G$ having both end nodes in $U$ is denoted by $E[U]$. 
Furthermore, for a graph $G$ its node set is denoted by $V[G]$ and its edge set by $E[G]$.
For two graphs $G$ and $\bar G$ we call $\bar G$ a subgraph of $G$ if
$V[\bar G] \subseteq V[G]$ and $E[\bar G] \subseteq E[G]$.
$G[H]:=(V[H],H)$ and $G[U]:=(U,E[U])$ refer to the subgraphs of $G$ induced by $H$ and by $U$, respectively. 
$G-H$ represents the graph obtained from $G$ by deleting all edges in $H$, while $G-U$ is used for
the graph obtained from $G$ by deleting all nodes of $U$ and all edges incident to some node in $U$.
For two graphs $G^\prime:=(R^\prime\dcup T^\prime, E^\prime)$ and
$G^{\prime\prime}:=(R^{\prime\prime}\dcup T^{\prime\prime},
E^{\prime\prime})$, their union defined as $\big(\,(R^\prime \cup
R^{\prime\prime}) \dcup (T^{\prime}\cup T^{\prime\prime}),\,
E^\prime\cup E^{\prime\prime}\,\big)$ is denoted by
$G^\prime+G^{\prime\prime}$. \\
A subgraph $\bar G$ of $G$ is called \emph{spanning} if each node of $G$ is
incident to at least one edge of $\bar G$.

\subsubsection{Problem setting}

The assignment problem has a well-known interpretation as a \emph{bipartite matching problem} in the graph $G=(R\dcup T, E)$.
It is, hence, natural to view \rap\ as a robust version of the bipartite matching problem, i.e.\ to find a minimum-cost subset $M\subseteq E$ such that, for every scenario $\Fset\in \uset$, the set $M\setminus\Fset$ contains a matching covering all nodes in $T$.
Moreover, if $|R|=|T|$ additionally holds, the problem becomes a robust variant of the \emph{perfect matching} problem. 

In the following, this point of view is adopted as it facilitates a clearer exposition of results and highlights an inherent connection between \rap\ and \emph{matching-covered graphs}, a notion that is repeatedly used in this paper to develop approximation algorithms.

The next statement shows that it suffices to consider \rap\ on balanced bipartite graphs. It also implies that feasibility conditions on \rap\ can be stated in terms of perfect matchings.

\begin{proposition}\label{prop:transformation-non-balanced-to-balanced-instances}
  Any \rap\ instance can be efficiently transformed to an equivalent weighted \rap\ instance with a balanced bipartite graph such that, for all $\alpha\geq 1$, any $\alpha$-approximation for the new instance can be used to efficiently construct an $\alpha$-approximation of the original instance.
\end{proposition}
\begin{proof}
  Consider any \rap\ instance on an unbalanced graph $G:=(R\dcup
  T,E)$ with $|T|< |R|$, uncertainty set $\uset$ and cost vector
  $c\in\R^{E}_{\geq 0}$.
  To transform $G$ into a balanced graph $G^\prime:=(R^\prime\dcup T^\prime,E^\prime)$,
  with $|R^\prime|= |T^\prime|$, a set $D$ of dummy task nodes of
  cardinality $|D|=|R|-|T|$ is introduced, and each $d \in D$ is
  connected with every resource node $r \in R$.  
  Further, let $E_D$ be the set of those newly introduced edges, and
  set $R^\prime := R$, $T^\prime := T \dcup D$, $E^\prime := E \dcup
  E_D$ and $\uset^\prime := \uset$.   
  Finally, choose the new cost vector $c^\prime\in\R^{E^\prime}_{\geq
    0}$ such that $c^\prime_e=0$, if $e \in E_D$, and $c^\prime_e=c_e$, for
  all $e\in E^\prime\setminus E_D=E$.  
  This procedure can be performed in polynomial time and leads to the
  desired instance on a balanced graph $G^\prime$.\\
  Now, let $V[E_D]$ denote all nodes from $G^\prime$ covered by $E_D$.  
  As $T \cap V[E_D] = \emptyset$, there exists a matching $M
  \subseteq E$ in $G$ that covers $T$ if and only if there is a perfect matching $M^\prime$ in $G^\prime$. 
  To see this, note that any matching $M\subseteq E$ covering $T$ can be extended to a perfect matching in $G^\prime$ by adding some new edges from $E_D$ while any perfect matching $M^\prime$ in $G^\prime$ must contain a matching $M\subseteq E$ covering $T$. 
  As the new edges from $E_D$ have zero costs, it follows that $c(M) = c^\prime(M^\prime)$. 
  This implies that the transformation preserves quality in terms of approximation.
\end{proof}
It is important to note that the transformation considered in 
Proposition~\ref{prop:transformation-non-balanced-to-balanced-instances} 
produces a \emph{weighted} instance with a balanced bipartite graph, even 
if the original instance had unit weights. We will hence not be able 
to use this result in a black-box fashion for our approximation algorithm 
for the unweighted case. We bring the details in Section~\ref{sec:cardrap-is-hard-to-approximate}.

The discussion above justifies to focus on \rap\ defined on balanced bipartite graphs.   
From now on, it is further assumed that each failure scenario $\Fset\in \uset$ is 
composed of a single edge, i.e.\ $|\Fset|=1$. We can thus henceforth assume that $\uset$ 
is simply the set of vulnerable edges, where the scenarios correspond to the failure
of any single edge of $\uset$. As we show in this paper, this special case of \rap\
is already interesting both from the application, and the algorithmic points of view.

Summing up, the specific variant of \rap\ considered in this paper is the following. 
\begin{problem*}[The Robust Assignment Problem (\rap)]
\text{}
\begin{itemize}
\item \underline{Input:} 
  Tuple $(G,\uset,c)$, where $G := (R\dcup T,E)$ is a balanced, bipartite graph, i.e.\  $|R|=|T|$, $\uset \subseteq E$ is a set of vulnerable edges, and $c\in\R_{\geq 0}^E$ is a non-negative cost vector. 
\item 
  \underline{Output:} If exists, an optimal solution for
  \begin{align}
    \tag{\rap}
    \label{eq:weighted-MRPMP}
    \begin{array}{ll}
      \min & c(\rset)  \\ 
      \textnormal{s.t.} & 
      \forall \fset\in \uset: 
      \rset \setminus \{ \fset \} \text{ contains a perfect matching
        in } G\\
      & \rset \subseteq E .
    \end{array}
  \end{align}
\end{itemize}
\end{problem*}
For a given instance of \rap\, $n$ and $m$ denote the number of nodes 
and edges of the underlying graph $G$, i.e.\ $n:=|R|+|T|$ and $m:=|E|$.

Two special cases of \rap\ are of
particular interest. The first case, denoted by \urap, is given when
the cost function of \rap\ is \emph{unweighted}, i.e.\ when the task
is to find a feasible (robust) solution of minimum cardinality. 
In the second case, every edge in the underlying graph $G=(R\dcup T,
E)$ is assumed to be vulnerable, i.e.\ $\uset = E$. In this case, the \rap\ instance is called \emph{uniform}.

Before we move on, we briefly treat the problem of deciding whether a \rap\ instance is feasible.
Observe that, for an arbitrary instance $\inst=(G,\uset,c)$ of \rap,
the feasible set is \emph{monotonic} in the sense
that any superset of a feasible solution is feasible as well. 
Thus, $\inst$ is a feasible instance if and only if the edge 
set $E$ of $G$ is a feasible solution, or equivalently, $E\setminus\{\fset\}$ 
contains a perfect matching in $G$, for every $\fset\in\uset$.  
The latter condition can be checked using any polynomial algorithm 
for finding maximum matchings in bipartite graphs.
Therefore, feasibility of \rap\ can be efficiently verified.
We will henceforth assume that any \rap\ instance considered in this 
paper is feasible.

It is worth to remark that the latter is no longer true when the uncertainty
set $\uset$ is given implicitly.
For example, consider a balanced bipartite graph $G:=(R\dcup T, E)$
and an uncertainty set $\uset:= \{\Fset \subseteq E \colon
|\Fset|=k\}$, for some $k\in\Z_{>1}$, presented by the parameter $k$. 
Then, checking for every $\Fset\in\uset$ whether $E\setminus\Fset$ contains 
a perfect matching in $G$ is equivalent to the problem of deciding if 
the so-called \emph{matching preclusion number} of $G$ is at most $k$. 
The latter problem was proved to be \np-complete for bipartite graphs by 
Dourado et al.\ in~\cite[Thm. 2]{dourado_et_al_15}.

\subsubsection{Matching-Covered Graphs}
\label{subsub_mcg}
The algorithmic results derived in this paper rely on a tight connection between 
\rap\ and \emph{matching-covered} graphs, a well-known notion 
in matching theory. Recap that a graph is \emph{matching-covered} if each of its edges appears in some perfect matching\,\footnote{The notion of matching-covered graphs is originally introduced for connected graphs. In this paper we use this term also for disconnected graphs.
Note further that some authors use synonymously the notion \emph{$1$-extendable} or, in the bipartite case, \emph{elementary} (cf.~\cite{lovasz_plummer_book_86}).}.

It turns out that inclusion-wise minimal solutions of any \rap\ instance are matching-covered as the following proposition states. 

\begin{proposition}\label{prop:matching_covered_and_RAP}
  Let $\inst:=(G,\uset,c)$ be any feasible \rap\ instance.
  Then, a subset $\rset$ of edges in $G$ 
  is an inclusion-wise minimal feasible solution to $\inst$ if and only if its induced subgraph $G[\rset]$ spans $G$ and is inclusion-wise minimal
with the properties of being matching-covered and that every isolated edge $e\in G[\rset]$ is not vulnerable, i.e.\ $e\not\in \uset$.
\end{proposition}
\begin{proof}
  \emph{``only if'' part. }
  Let $\rset$ be an inclusion-wise minimal feasible solution to $\inst$, and let $e\in\rset$ be any edge.  
 Firstly, if $e$ does not appear in any perfect matching, $\rset\setminus \{e\}$ is feasible to $\inst$.  
Thus, inclusion-wise minimality of $\rset$ implies that $G[\rset]$ is matching-covered.   
Secondly, if $f\in\uset$ is an isolated edge in $G[\rset]$ then $\rset\setminus\{\fset\}$ cannot contain a perfect matching in $G$.
This contradicts that $\rset$ is feasible.   
Thirdly, assume that $G[\rset]$ contains a proper spanning subgraph $G[\rset^\prime]$, induced by some $\rset^\prime \subsetneq \rset$, that is matching-covered and that has no isolated edges from $\uset$. 
Consider an arbitrary $\fset\in\rset^\prime\cap\uset$. As $\fset$ is not isolated in $G[\rset^\prime]$, there is an $e^\prime \in\rset^\prime$ adjacent to $\fset$. 
Since $G[\rset^\prime]$ spans $G$ and is matching-covered, there is a perfect matching in $G$ that contains $e^\prime$ and that does not contain $\fset$. 
This, however, shows that $\rset^\prime$ is also feasible to $\inst$  contradicting the inclusion-wise minimality of $\rset$.

\emph{``if'' part. }
Let $\rset$ be a subset of edges from $G$ such that its induced subgraph $G[\rset]$ spans $G$ and is inclusion-wise minimal with respect to the properties of being matching-covered and that every isolated edge in $G[\rset]$ is not vulnerable. 
Showing that $\rset$ is feasible to $\inst$ is similar to the proof of the feasibility of $\rset^\prime$ in the first direction of the proof.
Assume now that there is an $\bar\rset \subseteq \rset$, $\bar\rset \neq \rset$, that is feasible to $\inst$. This firstly implies that the induced subgraph $G[\bar\rset]$ spans $G$. Secondly, no vulnerable edge $f\in\bar\rset\cap\uset$, if exists, can form an isolated edge in $G[\bar\rset]$.
Moreover, it can be assumed that $G[\bar\rset]$ is also matching-covered, as otherwise, each $e\in \bar\rset$ not extendable to a perfect matching in $G[\bar\rset]$ can be removed from $\bar\rset$. This way, a proper subset of $\rset$ is obtained whose induced subgraph spans $G$ and is matching-covered. This concludes the proof.
\end{proof}

Proposition~\ref{prop:matching_covered_and_RAP} provides a very useful characterization of inclusion-wise minimal solutions of \rap, as it allows to make use of various results on matching-covered graphs to design algorithms for \rap. 
In particular, it allows to identify feasible subgraphs and augment them to feasible solutions for the entire instance  by adding structures that maintain the property of being matching-covered.

\subsubsection{Results for RAP}
One important contribution of this paper is to show that approximating \rap\ is as 
hard as to approximate the well-known Set Cover Problem. This is stated in the following theorem. 
\begin{theorem}\label{thm:hardness_RAP_uniform}
\rap\ admits no polynomial $d\log n$-approximation algorithm for any $d < 1$, unless 
\np\ $\not\subseteq \dtime (n^{\log \log n})$\footnote{Recap that $\np \subseteq \dtime (n^{\log \log n})$ would imply 
the existence of quasi-polynomial time algorithms for \np-hard problem.}.
This is true even for uniform \rap.
\end{theorem}

Theorem~\ref{thm:hardness_RAP_uniform} motivates developing an approximation 
algorithm for \rap\ with the matching asymptotic bound $O(\log n)$,
imposed by the previous theorem. 
We achieve this goal in the next theorem.

\begin{theorem}~\label{thm:RAP-admits-log-n-approximation-algorithm}
\rap\ admits a randomized polynomial $O(\log n)$-approximation algorithm.
\end{theorem} 
For the proof of Theorem~\ref{thm:RAP-admits-log-n-approximation-algorithm}, an algorithm with the desired approximation quality is presented.
The algorithm constructs iteratively a solution maintaining the invariant that, at any iteration, the edges selected so far, form a matching-covered graph. 
It is, however, unclear how to arrive at the desired approximation for \rap\ when only properties of matching-covered graphs are taken into account. 
Therefore, the latter technique is combined with additional tools from the theory of linear programming (\lp) and randomized rounding. 
More precisely, the algorithm starts with solving an \lp\ relaxation of \rap, derived from a natural integer linear programming (\ilp) formulation of the problem.
The fractional solution obtained this way is used to guide an iterative randomized procedure. 
In each iteration a fractional bipartite matching corresponding to part of the fractional solution is selected.
A decomposition of this fractional matching into a convex combination of integral matchings is then used to randomly pick one matching, and a carefully selected subset of this matching is added to the current solution. 
To bound the quality, it does not suffice to bound the number of iterations, or the expected number of times an edge is part of a candidate matching. 
Instead, a discharging argument, that assigns costs to nodes depending on the graph selected so far, is used.

A detailed presentation of all results for \rap, including the proofs, is given in Section~\ref{sec:rap-hardness-and-approximability}.

\subsubsection{Results for card-RAP}

Our main complexity result for the unweighted \rap\
states that it is \np-hard to approximate within some constant 
$\delta>1$, even for uniform uncertainty sets.
\begin{theorem}
\label{thm:card-rap-is-hard-to-approximate}
For some constant $\delta > 1$, there is no polynomial $\delta$-approximation
for uniform \urap, unless \p $=$ \np.
\end{theorem}

On the positive side, we are able to use the strong connection between 
\rap\ and matching-covered graphs to develop a constant factor approximation 
algorithm for \urap\ using so-called \emph{ear decompositions}. 
\begin{theorem}
\label{thm:approximation_unweighted_RAP}
 \urap\ admits a polynomial $1.5$-approximation algorithm in the uniform case, 
and a $3$-approximation algorithm in the non-uniform case.
\end{theorem}
The algorithm starts by producing an ear decomposition of the input graph. Then, 
it selects a certain subset of the edges to be part of the solution, by
processing the ears in the decomposition in the order given by the decomposition,
and omitting the edges corresponding to ears of length one.

Theorems~\ref{thm:card-rap-is-hard-to-approximate} and~\ref{thm:approximation_unweighted_RAP} imply 
that, assuming $\p \neq \np$, the true approximability thresholds for uniform \urap\ and \urap\ lie in the 
intervals $[\delta,1.5]$ and $[\delta^\prime,3]$, for some $\delta,\delta^\prime>1$.
These results are proved in Section~\ref{sec:card-rap-hardness-and-approximation}.

\bigskip
To complete the complexity landscape of \urap, the case with only two vulnerable edges is also considered. 
This special case comprises the simplest possible variant of \rap\ that is not equivalent to a nominal assignment problem\,\footnote{Observe
that the case of a single vulnerable edge $\uset = \{\fset\}$ is solvable by reporting 
any minimum-cost perfect matching in the graph $(R\dcup T, E\setminus \{\fset\})$ as a solution.}.
It turns out that this special case is  already \np-hard, thus drawing a sharp threshold for polynomial solvability of \rap.

\begin{theorem}\label{thm:hardness_unweighted_RAP_two_scenarios}
 \urap\ is \np-hard even when restricted to instances with two vulnerable
 edges, i.e.\ with $\uset= \{f_1, f_2\}$.
\end{theorem}
The proof of Theorem~\ref{thm:hardness_unweighted_RAP_two_scenarios} relies on a connection to a new \np-hard problem called Shortest Nice Path Problem (\snpp).
The goal is to partition a graph into a path connecting two given nodes and a matching such that their union covers all nodes, so as to minimize the length of the path.
This problem might be interesting in its own right.
The formal definition of SNPP and the proof of 
Theorem~\ref{thm:hardness_unweighted_RAP_two_scenarios} is presented in
Section~\ref{sec:proof-sec-singleton-two-scenarios-hardness}.\\

\begin{remark}
\label{rem:hardness-of-robust-problems-with-constant-number-of-uncertain-resources}
To the best of the authors' knowledge, this is the first example of an \np-hard 
robust counterpart of a polynomial optimization problem, with a constant 
number of vulnerable resources.
Note that there are many examples of optimization
problems that become \np-hard when the robust counterpart is allowed to
contain a constant number of scenarios (see e.g.~\cite{kouvelis_yu_97}).
However, in all such examples, each scenario affects a \emph{non-constant}
number of resources.
\end{remark}


\section{Related work}
\label{sec:related-work}  

Redundancy-based robustness is a paradigm that motivates many
well-studied problems, including the minimum $k$-connected subgraph problem~(see Gabow et al.~\cite{GabowGoemansTardosWilliamson},
Cheriyan, Seb{\H{o}} and Szigeti~\cite{cheriyan_et_al_01}, and Seb{\H{o}} and Vygen \cite{sebHo2014shorter}), survivable and robust network design problems~(Jain~\cite{Jain}, Chekuri~\cite{chekuri2007routing}, Adjiashvili, Stiller and Zenklusen~\cite{AdjishviliStillerZenklusen2014}, and Adjiashvili~\cite{adjiashvili_bulk_planar}), robust facility location problems~(Jain and Vazirani~\cite{jain2000approximation}, and Smamy and Shmoys~\cite{swamy2003fault}), robust spanner problems~(Chechik et al.~\cite{chechik2009fault}, and Dinitz and Krauthgamer \cite{dinitz2011fault}), and many more. 
All of the latter models bare a close resemblance to \rap: they assume resources to be vulnerable and ask to find a minimum-cost set of resources that contains a desired structure even in case any vulnerable resource, or set of resources, fails.

A relatively new approach to redundancy-based robustness is the incorporation of non-uniform uncertainty sets~(\cite{AdjishviliStillerZenklusen2014,adjiashvili_bulk_planar}). 
\rap\ is seen as a robust model of this type, as both vulnerable and
invulnerable edges are allowed to appear in the same instance.

The study of robustness with respect to cost uncertainty is
initiated by Kouvelis and Yu~\cite{kouvelis_yu_97}, 
and Yu and Yang~\cite{yu_yang_98}. 
These works mainly consider the min-max model, where the goal
is to find a solution that minimizes the worst-case cost according to the given set of cost functions.
A survey on this topic is given by Aissi, Bazgan and Vanderpooten~\cite{Survey_AissiBazgenVanderpooten}.
A closely related class of multi-budgeted problems is received considerable 
attention (see e.g.~Grandoni et al.~\cite{grandoni2014new} and references therein).
The latter work includes variants of the related multi-objective matching problem 
for which a polynomial time approximation scheme is presented.

Various variants of the robust matching problems are considered in the literature.
Hassin and Rubinstein~\cite{hassin2002robust}, and Fujita, Kobayashi, and Makino~\cite{fujita2010robust} study the following notion of an $\alpha$-robust matching.
A perfect matching $M$ in a weighted graph is $\alpha$-robust (for $\alpha \in (0,1]$), if for every $p \leq |M|$, the $p$ heaviest edges of the matching have total weight at least $\alpha$ times the weight of a maximum-weight matching of size $p$. 
In~\cite{hassin2002robust} the authors prove that the complete graph $K_n$ contains a $\frac{1}{\sqrt{2}}$-robust matching and this bound is tight in general.
Additionally, the authors provide a polynomial-time algorithm to find such a matching.
Building upon these results Fujita, Kobayashi, and Makino~\cite{fujita2010robust} prove that the problem of deciding 
if the input graph has $\alpha$-robust matching with $\alpha \in (\frac{1}{\sqrt{2}},1)$ 
is \np-complete, and extend the original algorithm to the matroid intersection problem.

Deineko and Woeginger~\cite{deineko2006robust} show that the min-max-robust assignment problem with a fixed number of scenarios is equivalent to the \emph{exact perfect 
matching problem}, a famous problem with unknown complexity status. 
In the case of a variable number of scenarios the min-max-robust problem is \np-hard, as proved by Aissi, Bazgan and Vanderpooten~\cite{aissi2005complexity}.

Laroche et al.~\cite{laroche_et_al_14} investigate the following robust variant of the matching problem.
Find the maximum number of nodes that can be removed from the larger side of the bipartition such that all nodes from the opposite bipartition can still be matched.
The authors also provide a polynomial-time algorithm for this problem.

Bertsimas, Brown and Caramanis~\cite{BertsimasBrownCaramanis} provide a comprehensive survey of the different 
facets of robust optimization. Surveys on robust combinatorial optimization are given in the PhD theses 
of Olver~\cite{thesis_olver} and Adjiashvili~\cite{thesis_adjiashvili}.

Brigham et al.~\cite{brigham_et_al_05} study the minimum number of edges to be removed from a graph to arrive at a graph without a perfect matching, i.e.\ the matching preclusion number.
For several graph classes the matching preclusion number as well as all optimal solutions are computed.
Cheng et al.~\cite{cheng_et_al_09} extend this concept by excluding the obvious case when the edges' removal produces isolated nodes introducing the \emph{conditional matching preclusion sets}.
For several basic graph classes optimal such sets are presented.
Dourado et al.~\cite{dourado_et_al_15} study the closely related problem of finding a 
robust and recoverable matching: Given a graph $G$, a set $F$ of edges, and two 
integers $r$ and $s$, does the graph $G$ have a perfect matching $M$ such that, for 
every choice $F'$ of $r$ edges from $F$, the graph $G-F'$ contains a perfect 
matching $M'$ having a symmetric difference with $M$ of size at most $s$?
Such a matching $M$ is called \emph{$r$-robust} and \emph{$s$-recoverable}.
The authors prove hardness of several related problems and present some tractable cases.

Chegireddy and Hamacher~\cite{chegireddy_k_best_matchings_1987} investigate the problem 
of finding the $k$ best perfect matchings in terms of weight in a given graph and provide an $O(kn^3)$ algorithm.

Hung, Hsu and Sung~\cite{hung_et_al_93} consider the so-called \emph{most vital edges} 
with respect to a weighted bipartite matching, i.e.\ the edges causing the largest decrement in the value 
of the maximum matching upon removal and provided an $O(n^3)$ algorithm.
Zenklusen~\cite{zenklusen2010matching} analyzes the related interdiction problem, i.e.\ the 
problem of choosing a set $R$ of edges (or nodes) in some weighted graph $G$ respecting 
a given budget in order to minimize the weight of maximum matching in subgraph of $G$ resulting by removal of $R$.
Several special cases of that problem are shown to be \np-hard. An $O(1)$-approximation algorithm is 
also provided.

Darmann et al.~\cite{darmann_et_al_11} study the hardness of the maximum matching 
problem with the additional structure of conflict (respectively, forcing) graphs.
Such graphs describe pairs of edges such that at most (respectively, at least) one 
of the two is to be part of the solution. The problem is shown to be \np-hard 
even for very restricted classes of conflict and forcing graphs.
Öncan, Zhang and Punnen~\cite{oncan_et_al_13} presents 
further complexity results and heuristics for the perfect matching problem 
with conflict constraints.

\medskip
Structural and graph-theoretic aspects of modifications of the matching
problem have also been considered in the literature.

Plesník~\cite{plesnik_72} proves that an $(r-1)$-edge-connected $r$-regular graph remains perfectly matchable after removing $r-1$ arbitrary edges. 
In the context of \rap\ this means that a $2$-regular edge-connected subgraph of a (non-)bipartite graph is feasible in the uniform case.

Some generalizations of matching-covered graphs, such as \emph{$n$-extendable} graphs, were investigated in the literature.
A graph is called \emph{$n$-extendable} if any matching of size $n$ can be extended to a perfect matching. Thus, matching-covered graphs form the class of $1$-extendable graphs. 
Moreover, $n$-extendability can be seen as a special case of the so-called \emph{$E(m,n)$--property}.  
A graph satisfies the \emph{$E(m,n)$--property} if, for any pair $(M,N)$ of disjoint matchings with $|M|=m$ and $|N|=n$, the graph 
has a perfect matching that contains $N$ and does not contain any element of $M$.
Such graphs were studied by several authors, e.g.\ Porteous and Aldred~\cite{porteous_aldred_96}, Aldred et al.~\cite{aldred_et_al_99}, and Aldred and Plummer~\cite{aldred_plummer_99}.

Wang, Yuan and Zhou~\cite{wang_et_al_09} considered $k$-edge-deletable IM-extendable graphs that are characterized as follows. 
A graph $G$ has the latter property if after removing any set $F$ of $k$ edges from 
it, every induced matching $M$ of $G-F$ (i.e.\ a matching such that $E[V[M]] = M$) is 
contained in a perfect matching of $G-F$.

Several publications deal with algorithms to identify all edges of a graph being part of some maximum matching. 
Regin~\cite{regin_94}, Tassa~\cite{tassa_12}, and Costa~\cite{costa_94} provided results on bipartite graphs, 
while Carvalho and Cheriyan~\cite{carvalho_cheriyan_05} considered general graphs.
For bipartite graphs, Valencia and Vargas~\cite{valencia_vargas_16} studied 
the weighted version of this problem that asks to find all edges belonging to 
some minimum weight perfect matching.


\section{Hardness and approximability of RAP} 
\label{sec:rap-hardness-and-approximability}

This section deals with the approximability of \rap.
In Subsection~\ref{sec:hardness_uniform_basic_transformation}, it is first shown that the well-known \emph{Set Cover Problem} can be polynomially reduced to weighted \rap. 
From that reduction,  it is then concluded that under a mild and widely accepted assumption the bound on a reachable approximation guarantee on \rap\ is of order $\log n$, proving Theorem~\ref{thm:hardness_RAP_uniform}.  
Secondly, an asymptotically tight randomized algorithm is developed in Subsection~\ref{sec:algorithm_RAP}.
The existence of this algorithm implies the correctness of Theorem~\ref{thm:RAP-admits-log-n-approximation-algorithm}.

\subsection{Reduction from the Set Cover Problem}
\label{sec:hardness_uniform_basic_transformation}

The proof of the hardness of the approximability of \rap\ relies on a reduction from the Set Cover Problem\footnote{For the underlying \np-complete decision problem, see~\cite[Problem SP5]{garey_johnson_79}.}.   
\begin{problem*}[{Set Cover Problem (\setcover)}]
\label{prob:set-cover}
\text{}
\begin{itemize}
\item \underline{Input:} Tuple $([k], \mS)$, where $[k] = \{1,\dots, k\}$ is a finite ground set and $\mS:=\{S_1, \dots, S_l\}$ is a collection of subsets of $[k]$, for some $k,l\in\Z_{\geq 1}$.    
\item 
  \underline{Output:} A collection $\mC \subseteq \{S_1, \dots, S_l\}$ of minimum cardinality such that
  that $\bigcup_{S \in \mC} S = [k]$ holds. 
\end{itemize}
\end{problem*}

For a given instance $([k],\mS)$, the existence of any cover for ground
$[k]$ can efficiently be verified, simply by checking if $\bigcup_{S \in \mS} S = [k]$ holds. 
From now on, it is hence assumed that any considered \setcover\
instance contains at least one feasible cover solution.

For any feasible \setcover\ instance  $([k],\mS)$ consider now the balanced bipartite graph $G:=(R\dcup T, E)$ illustrated in Figure~\ref{fig:graph-corresponding-to-an-instance-of-set-cover} and constructed as follows.
\begin{figure}[htb]
  \centering
  \begin{tikzpicture}[inner sep = 2pt]
    \node (0) at (-0.5, 0) [circle, dashed, draw] {$\bar{u}_k$};
    \node (1) at (-0.5, 3) [circle, dashed, draw] {$\bar{u}_2$};
    \node (2) at (-0.5, 4) [circle, dashed, draw] {$\bar{u}_1$};
    
    \node (3) at (1.5, 0) [circle, draw] {$u_k$};
    \node (4) at (1.5, 3) [circle, draw] {$u_2$};
    \node (5) at (1.5, 4) [circle, draw] {$u_1$};
    
    \node (6) at (3.5, 0.5) [circle, draw] {$v_{S_l}$};
    \node (7) at (3.5, 2.5) [circle, draw] {$v_{S_2}$};
    \node (8) at (3.5, 3.5) [circle, draw] {$v_{S_1}$};
    
    \node (9) at (5, 0.5) [circle, draw] {$\bar{v}_{S_l}$};
    \node (10) at (5, 2.5) [circle, draw] {$\bar{v}_{S_2}$};
    \node (11) at (5, 3.5) [circle, draw] {$\bar{v}_{S_1}$};
    
    \node (12) at (7, 0.5) [circle, draw] {$\tilde{v}_{S_l}$};
    \node (13) at (7, 2.5) [circle, draw] {$\tilde{v}_{S_2}$};
    \node (14) at (7, 3.5) [circle, draw] {$\tilde{v}_{S_1}$};
    
    \node (15) at (8.5, 0.5) [circle, draw] {$w_{S_l}$};
    \node (16) at (8.5, 2.5) [circle, draw] {$w_{S_2}$};
    \node (17) at (8.5, 3.5) [circle, draw] {$w_{S_1}$};
    
    \node (18) at (10.5, 0) [circle, dashed, draw] {$\bar{u}_k$};
    \node (19) at (10.5, 3) [circle, dashed, draw] {$\bar{u}_2$};
    \node (20) at (10.5, 4) [circle, dashed, draw] {$\bar{u}_1$};
    
    \draw[-, thick] (17) -- (20);
    \draw[-, thick] (17) -- (19);
    \draw[-, thick] (16) -- (18);
    \draw[-, thick] (15) -- (18);
    \draw[-, thick] (15) -- (19);
    \draw[-, thick] (15) -- (20);
    
    \node at (0.5,-0.8) {$E_1$};
    \node at (2.5,-0.8) {$E_2$};
    \node at (4.25,-0.8) {$E_3$};
    \node at (6,-0.8) {$E_4$};
    \node at (7.75,-0.8) {$E_5$};
    \node at (9.5,-0.8) {$E_6$};
    
    \draw[-, thick] (3) -- (7);
    \draw[-, thick] (3) -- (6);
    \draw[-, thick] (4) -- (8);
    \draw[-, thick] (4) -- (6);
    \draw[-, thick] (5) -- (6);
    \draw[-, thick] (5) -- (8);
    
    \draw[-, thick] (12) -- (15);
    \draw[-, thick] (13) -- (16);
    \draw[-, thick] (14) -- (17);
    
    \foreach \i in {0,..., 6}
    {
      \pgfmathtruncatemacro{\head}{\i * 3}
      \pgfmathtruncatemacro{\tail}{\head + 1}
      \draw[-, very thick, loosely dotted] (\head) -- (\tail);	
    }
    
    \foreach \i in {0,..., 2}
    {
      \pgfmathtruncatemacro{\head}{\i}
      \pgfmathtruncatemacro{\tail}{\head + 3}
      \draw (\head) edge [-,thick] node [auto] {$\in F$} (\tail);
}

\foreach \i in {0,..., 2}
{
  \pgfmathtruncatemacro{\head}{\i + 6}
  \pgfmathtruncatemacro{\inner}{\head + 3}	
  \pgfmathtruncatemacro{\tail}{\head + 6}
  \draw[-, thick] (\head) -- (\inner);	
  \draw (\inner) edge [-,thick] node [auto] {$c_e = 1$} (\tail);	
}

\end{tikzpicture}
\caption{The figure shows the graph $G$ that is constructed from an instance $([k],\{S_1,\hdots,S_l\})$ of \setcover.
For the purposes of legibility, the dashed nodes representing node set $\{\bar{u}_s \mid s\in [k]\}$ appear twice in the figure.
}	 
\label{fig:graph-corresponding-to-an-instance-of-set-cover}
\end{figure}
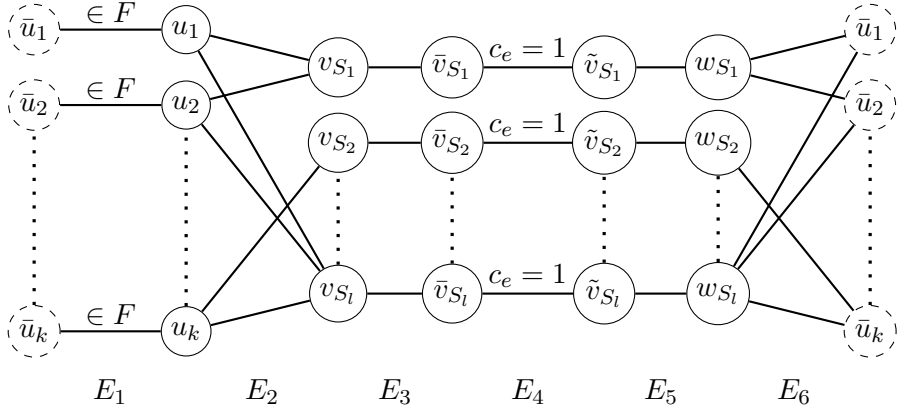
\begin{itemize}
\label{tsteps-one-to-four}
\item[{\small\bf (T1)}]
  For each $s\in [k]$, a node $u_s$ is introduced and added to $T$, and for each $S\in \mS$, a node $v_S$ is introduced and added to $R$. 
  Furthermore, $G$ contains the edge $\{u_s, v_S\}$ if and only if $s\in S$. All these edges form subclass $E_2$.
\item[{\small\bf (T2)}]
  For each $s\in [k]$, a copy $\bar{u}_s$ of node $u_s$ is introduced and added to $R$, and $u_s$ and $\bar{u}_s$ are connected by an edge. 
  These edges are composited to $E_1$.
\item[{\small\bf (T3)}]
  For each $S \in \mS$, two copies $\bar{v}_S$ and $\tilde{v}_S$ of node $v_S$ are introduced whereupon 
  the first copy is added to $T$, and the second one is added to $R$.
  The original node $v_S$ is connected with $\bar{v}_S$, and $\bar{v}_S$ is further connected with $\tilde{v}_S$. 
  All edges of the form $\{v_S,\bar{v}_S\}$ yield subclass $E_3$, while all edges of type $\{ \bar{v}_S, \tilde{v}_S\}$ belong to subclass $E_4$. 
\item[{\small\bf (T4)}]
  For each $S\in \mS$, a fourth type of node is introduced that is denoted by $w_S$ and added to $T$. 
  Each node $w_S$ is connected with the corresponding node $\tilde{v}_S$. These edges form subclass $E_5$.
  Finally, an edge between $w_S$ and $\bar{u}_s$ is introduced and added to subclass $E_6$ if and only if $s\in S$.  
\end{itemize} 
With some $\uset\subseteq E$ and some $c\in\R^E_{\geq 0}$, the tuple $(G,\uset,c)$ then defines an instance of \rap.
Note that both $E_2$ and $E_6$ encode whether element $s\in [k]$ is contained in subset $S\in \mS$, or not. 
$E_3$ and $E_5$ ensure the feasibility of the \rap\ instance, while the edges in $E_4$ are used to indicate which elements from $\mS$ are chosen to cover the ground set $[k]$.

The next lemma highlights the close relation between \setcover\ and the \rap.
\begin{lemma}
\label{lem:set-cover-rap-feasibility-version-1}
Let  $\inst:=([k],\mS)$ be an instance of \setcover, and let $\inst^\prime:=(G,\uset,c)$ be the \rap\ instance with $G:=(R\dcup T, E)$ as constructed by applying steps {\small\bf (T1) -- (T4)}, uncertainty set $\uset=E_1$ and cost vector $c\in\R^E_{\geq 0}$ with $c_e=1$, if $e\in E_4$, and $c_e=0$, if $e\in E\setminus E_4$.
\begin{itemize}
\item[(a)] 
  $\inst^\prime$ can be determined in time polynomially bounded in the input size of $\inst$.
\item[(b)]
 Let $\rset \subseteq E$ be with $E \setminus E_4 \subseteq \rset$.
 Then, 
 $\rset$ is feasible to $\inst^\prime$ if and only if $\mathcal{C}_{\rset} := \{S \in \mS \,\mid\, \{\bar{v}_S,\tilde{v}_S\} \in \rset \}$ is feasible to $\inst$.
\item[(c)]
  $\inst$ has a cover of size $\val$ if and only if $\inst^\prime$ contains a feasible solution $\rset\subseteq\bar E$ with objective value $c(X)=\val$.
\end{itemize}
\end{lemma}
\begin{proof}
\begin{itemize}
\item[(a)]
  Let $\inst:=([k],\mS)$ be given. 
  From the construction of $G$, it follows that  
  $|T| = |R| = 2|\mS|+|k|$ and that  
  $|E| = 3|\mS| +|k| + 2\sum_{S\in \mS} |S|$
  holds, i.e.\ the size of $G$ is polynomially bounded by the size of $\inst$.
  Further, $E_1$ and $c$ are, by definition, polynomially bounded in
  the input size of $G$.
  Thus, $\inst^\prime=(G,E_1,c)$ can be determined in polynomial time w.r.t.\ the input size of $\inst$. 
\item[(b)]
\emph{``only if'' part.\quad}
Let $\rset$  be any feasible solution to $\inst^\prime$ with $E\setminus E_4 \subseteq\rset$, and let $s\in [k]$. 
To show that $s$ is contained in some set of $\mathcal{C}_{\rset}$, consider
the edge $\fset_s=\{\bar{u}_s,u_s\} \in\uset=E_1$. 
As $X$ is feasible to $\inst^\prime$, there is a perfect matching $M \subseteq \rset$ in $G$ with $\fset_s\not\in M$.
Since $f_s\not\in M$, node $u_s$ must be matched with some node from $\{v_{S_1},\dots, v_{S_l}\}$ 
by the corresponding edge from $E_2$, i.e.\ $\{u_s, v_S\} \in \rset$, for some $S \in \mS$ with $s\in S$. 
Edge $\{u_s, v_S\}$ also covers node $v_S$. Thus, the node $\bar v_S$ must be matched with $\tilde v_S$. 
This implies that $\{\bar v_S, \tilde v_S\} \in \rset$ and, hence, that $S\in \mathcal{C}_{\rset}$. 
It follows that $s$ is covered, and that $\mathcal{C}_{\rset}$ is a
feasible cover for $\inst$.

\emph{``if'' part.\quad}
Let $\mC\subseteq \mS$ be a feasible cover for $\inst$. Then, define
$${\rset}:=(E\setminus E_4)\;\bigcup\; \big\{\,\{\bar{v}_S,\tilde{v}_S\}\mid S\in\mathcal{C}\big\},
\textrm{ implying that } \mathcal{C}=\mathcal{C}_X \textrm{  holds}.$$
Recap further that $X$ is feasible to $\inst^\prime$ if and only if $\rset \setminus \{f_s\}$ contains a perfect  matching of $G$, for all $\fset_s =\{\bar{u}_s,u_s\} \in\uset=E_1$.\\
To verify the correctness of the latter condition, let $\fset_s
\in\uset$ be arbitrary. 
Note also that $M_0:=E_1\cup E_3 \cup E_5$  is a perfect matching in $G$ with $f_s\in M_0\subseteq X$.
As $\mathcal{C}$ is a feasible cover for $\inst$, there exists a set $S \in \mS$ with $s\in S$ and $\{\bar v_S, \tilde v_S\} \in \rset$.
Now consider the cycle $C_0$ defined on node set $\{ \bar u_s, u_s, v_S, \bar v_S, \tilde v_S, w_S \}$. 
Observe that all edges in $C_0$ are contained in $X$. In particular, it holds that
$\{\bar u_s, u_s\}, \{v_S, \bar v_S\}, \{\tilde v_S, w_S\} \in M_0$ and that 
$\{u_s, v_S\}, \{\bar v_S, \tilde v_S\}, \{w_S, \bar u_s\}  \in \rset \setminus M_0$, i.e.\ $C_0$ forms an $M_0$-alternating cycle in $\rset$. 
Therefore, the set $M$ can be expressed as the symmetric difference $M_0 \triangle C_0$ and  is given as
$$\Big( M_0 \setminus \big\{ \{\bar u_s, u_s\}, \{v_S, \bar v_S\}, \{\tilde v_S, w_S\} \big\}\Big)\;\cup\;\big\{ \{u_s, v_S\}, \{\bar v_S, \tilde v_S\}, \{w_S, \bar u_s\} \big\}.$$
Hence, $M$ is a perfect matching in $G$ with $\{\bar u_s, u_s\}=\fset_s\not\in M \subseteq X$. 
\item[(c)] 
  This follows directly from part (b) and by the choice of $c$. 
\end{itemize}
\end{proof}

Lemma~\ref{lem:set-cover-rap-feasibility-version-1} already implies that any \setcover\ instance can be transformed, in polynomial time, into an equivalent \rap\ instance. It can further be concluded that \rap\ is as hard to approximate as \setcover.
The proof of Theorem~\ref{thm:hardness_RAP_uniform}, however, requires a stronger result, namely that \setcover\ can be even equivalently transformed into an \emph{uniform} \rap\ instance.

Again, let $\inst:=([k],\mS)$ be any \setcover\ instance, and let
$G:=(R\dcup T, E)$, with $E:=\bigcup_{i=1}^6 E_i$, be the graph
constructed from $\inst$ by applying the transformation steps {\small\bf (T1)--(T4)} as described above.
In general, the uniform version of the \rap\ instance $\inst^\prime$
on $G$ is not equivalent to $\inst$.  
To see this, pick any edge $e=\{\tilde{v}_S,w_{S}\}\in E_5$ and assume that $e$ is vulnerable, i.e.\ $e$ belongs to $\uset$. 
Then, it immediately follows that the (unique) edge
$e^\prime:=\{\bar{v}_S,\tilde{v}_S\}$ in $E_4$, incident with $e$,
must be part of any feasible solution, as otherwise $\tilde{v}_s$
cannot be matched without using $e$.  
Therefore, the only subset $\rset\subseteq E$ that contains $E\setminus E_4$ and that is feasible to $\inst^\prime$ is given by $X=E$.  
In the light of Lemma~\ref{lem:set-cover-rap-feasibility-version-1}, this shows that
a uniform \rap\ instance on $G$ has precisely one feasible solution
that only represents the trivial feasible cover $\mS$. The same line of arguments also works when $e$ is chosen from $E_3$ instead from $E_5$. 
 
To overcome the described problem, the graph $G$ must be changed in an appropriate way. For this, the following, additional transformation step is performed on $G$.
\begin{itemize}
\label{tstep-five}
\item[{\small\bf (T5)}] Each edge $e = \{v,w\} \in E_3\cup E_5$ is replaced by a cycle of length six. 
That means, for each $e$ introduce four auxiliary nodes 
  $x_1^{(e)}, x_2^{(e)}$, $y_1^{(e)}, y_2^{(e)}$. 
  Then, remove each $e=\{v,w\}\in E_3\cup E_5$ from $G$, and add the edges  
  $\{v,x_1^{(e)}\}, \{x_1^{(e)},x_2^{(e)}\}, \{x_2^{(e)},w\}$ as well as
  $\{v,y_1^{(e)}\}, \{y_1^{(e)},y_2^{(e)}\}, \{y_2^{(e)},w\}$ to $G$ (see Figure~\ref{fig:cloning-and-subdividing-the-crucial-edges-for-the-uniform-approximation-proofs}, for an illustration).
\end{itemize} 
\begin{figure}[htb]
  \centering
  \begin{tikzpicture}[inner sep = 2pt]
    \begin{scope}[yshift=-2cm]
      
      \node (0) at (0, 0) [circle, draw] {$v$};
      \node (1) at (4.5, 0) [circle, draw] {$w$};
      \node (2) at (1.5, 0.5) [circle, draw] {$x_1^{(e)}$};
      \node (3) at (3, 0.5) [circle, draw] {$x_2^{(e)}$};
      \node (4) at (1.5, -0.5) [circle, draw] {$y_1^{(e)}$};
      \node (5) at (3, -0.5) [circle, draw] {$y_2^{(e)}$};
      
      \draw[-, thick] (0) -- (2) -- (3) -- (1);
 \draw[-, thick] (0) -- (4) -- (5) -- (1);
 
 \node at (-1,0) {$\mapsto$};
 \node (6) at (-3,0) [circle,draw] {$v$};
 \node (7) at (-2,0) [circle,draw] {$w$};
 \draw[-, thick] (6) -- (7);

 \end{scope}
 \end{tikzpicture}
 \caption{The Figure shows the substitution of an edge $e = \{v, w\}$ by two parallel paths of length three with the help of four auxiliary  nodes.}	 
 \label{fig:cloning-and-subdividing-the-crucial-edges-for-the-uniform-approximation-proofs}
\end{figure}
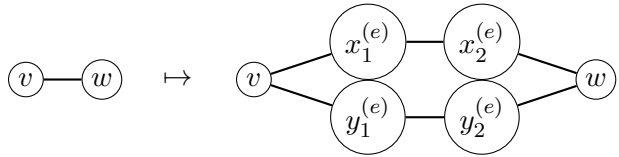
Transformation step {\small\bf (T5)} yields a new graph $\bar G :=
(\bar R \dcup \bar T, \bar E)$ with $\bar E = \bigcup_{j=1}^6 \bar
E_j$ and
\begin{itemize}
\item $\bar R := R\cup \{x^{(e)}_2, y^{(e)}_2 \mid e\in E_3\cup E_5\}$,\quad $\bar T := T\cup \{x^{(e)}_1, y^{(e)}_1 \mid e\in E_3\cup
  E_5\}$,
\item $\bar E_3:= \big\{\,\{v,x_1^{(e)}\}, \{x_1^{(e)},x_2^{(e)}\},
  \{x_2^{(e)},w\}, \{v,y_1^{(e)}\}, \{y_1^{(e)},y_2^{(e)}\}, \{y_2^{(e)},w\}
  \mid e=\{v,w\}\in E_3 \big\}$,
\item $\bar E_5:= \big\{\,\{v,x_1^{(e)}\}, \{x_1^{(e)},x_2^{(e)}\},
  \{x_2^{(e)},w\}, \{v,y_1^{(e)}\}, \{y_1^{(e)},y_2^{(e)}\}, \{y_2^{(e)},w\} \mid e=\{v,w\}\in E_5\big\}$,
\item $\bar E_1:= E_1$,  $\bar E_2:= E_2$,  $\bar E_4:= E_4$, and  $\bar E_6:= E_6$.    
\end{itemize}
Note that $\bar G$ is constructable in polynomial time (in terms of the input size of $\inst$). 
The next lemma shows that \setcover\ can be equivalently restated as a uniform \rap\ instance on $\bar G$.  
\begin{lemma}
\label{lem:set-cover-rap-feasibility-uniformversion}
Let $\inst:=([k],\mS)$ be an instance of \setcover, and let
$\bar\inst:=(\bar G,\bar E,\bar c)$ be the uniform \rap\ instance with
$\bar G:=(\bar R\dcup \bar T, \bar E)$ constructed by applying steps
{\small\bf (T1) -- (T5)}, and with $\bar c\in\R^{\bar E}_{\geq 0}$,
where $\bar c_e = 1$, if $e\in  \bar E_4$, and $\bar c_e=0$, otherwise. 
\begin{itemize}
\item[(a)] 
  $\bar \inst$ can be determined in time polynomially bounded in the input size of $\inst$.
\item[(b)]
  Let $\rset \subseteq \bar E$ such that $\bar E \setminus \bar E_4 \subseteq \rset$ holds.
  Then, 
  $\rset$ is feasible to $\bar \inst$ if and only if $\mathcal{C}_{\rset} := \{S \in \mS \,\mid\, \{\bar{v}_S,\tilde{v}_S\} \in \rset \}$ is feasible to $\inst$.
\item[(c)]
  $\inst$ has a cover of size $\val$ if and only if $\bar\inst$
  contains a feasible solution $\rset\subseteq\bar E$ with objective
  value $\bar c(X)=\val$.
\end{itemize} 
\end{lemma}
The proof of Lemma~\ref{lem:set-cover-rap-feasibility-uniformversion} is very similar to the proof of Lemma~\ref{lem:set-cover-rap-feasibility-version-1}, and is hence omitted here. With Lemma~\ref{lem:set-cover-rap-feasibility-uniformversion} at hand, Theorem~~\ref{thm:hardness_RAP_uniform} can be proved.

\begin{proof}[Proof of Theorem~\ref{thm:hardness_RAP_uniform}]
  By Lemma~\ref{lem:set-cover-rap-feasibility-uniformversion}, any
  \setcover\ instance $\inst:=([k],\mS)$ can be reduced, in polynomial
  time, into a uniform \rap\ instance $\bar\inst:=(\bar G, \bar E,
  \bar c)$, with $\bar G=(\bar R\dcup \bar T, \bar E)$ and $\bar c\in\R^{\bar E}_{\geq 0}$ as specified in the statement of Lemma~\ref{lem:set-cover-rap-feasibility-uniformversion}.
  In particular, it follows from part (c) that
  a feasible cover $\mC\subseteq \mS$ in $\inst$ has a cardinality of
  $\val\in\Z_{\geq 0}$ if and only if there is a feasible solution
  $\rset\subseteq \bar E$ for $\bar\inst$ containing $\val$ different
  edges out of  $\big\{ \{\bar v_S, \tilde{v}_S\}\mid S\in\mS \big\}
  (=\bar E_4=E_4)$. \\
  Feige~\cite{FeigeSetCover} showed that, for any $d < 1$,  \setcover\ admits no polynomial time  
  $d\log n$-approximation algorithm unless \np$\,\subseteq\, \dtime (n^{\log \log n}$). His result also holds when $|\mS|$ is polynomial in $k$, implying that $\bar G$ has size polynomial in $k$.\\
 Combining Feige's result with Lemma~\ref{lem:set-cover-rap-feasibility-uniformversion} completes the proof.
\end{proof}


\subsection{$O(\log n)$-Approximation for RAP}
\label{sec:algorithm_RAP}  

This section presents a polynomial $O(\log n)$-approximation 
algorithm for \rap, thus proving
Theorem~\ref{thm:RAP-admits-log-n-approximation-algorithm}. 
After introducing some basic notation, the approximation algorithm is first described for the
uniform case.
Afterwards, it is explained how the algorithm can be extended to the non-uniform case.  
  
Again, it is assumed that the \rap\ instance,
consisting of a balanced bipartite graph $G=(R\dcup T, E)$,
an uncertainty set $\uset\subseteq E$ and of an non-negative cost
vector $c\in\R^E_{\geq 0}$, is feasible. 
Furthermore, for each $S\subseteq E$, $\chi^S\in\{0,1\}^E$ represents the \emph{incidence vector} of $S$, 
i.e.\ the vector with $\chi^S_e = 1$ if $e\in S$ and $\chi^S_e = 0$, otherwise.

The algorithm is based on an \lp-rounding procedure that works with a
relaxation of the integer linear formulation of \rap. 
For this, let $P_G \subseteq \R^{E}$ be the perfect matching polytope associated
with $G$, i.e.\ $P_G$ is the convex hull of all incidence vectors of perfect
matchings in $G$. 

A standard \ilp\ formulation of \rap\ involves the following variables. 
\begin{enumerate}
  \setlength{\itemsep}{-0.25ex}
 \item [(i)] $x^{-f}\in\{0,1\}^E$ represen\-ting a perfect matching in $G-\{\fset\}$, for all $\fset\in\uset$,
 \item [(ii)] $y\in\{0,1\}^E$ encoding a feasible solution to \rap.
\end{enumerate}
Then, \rap\ can be modeled as an \ilp\ as follows.
\begin{align}
\tag{\ilp}
\label{eq_milpformulation}
\begin{array}{rrcll}
\min & c^\top y\\[1ex]
\textrm{s.\,t.} & x^{-f} & \in & P_G\cap\{x\in \R^{E} \,\mid\, x_f=0\}, & \textrm{
  for each } f\in\uset,\\
& y & \geq & x^{-f}, &  \textrm{ for each } f\in\uset, \\
& x^{-f} & \in & \{0,1\}^E, &  \textrm{ for each } f\in \uset, \\
& y & \in & \{0,1\}^E. &
\end{array}
\end{align}
The \lp-relaxation (\lp) is obtained by
relaxing all integrality constraints in~\eqref{eq_milpformulation}.
To keep notation short, let $x\in \big(\R^{E}\big)^{E}$  be the
vector with parts $x^{-f}$, $f\in\uset$. 
It is straightforward to verify that integer solutions
to~\eqref{eq_milpformulation} coincide with feasible solutions 
to the \rap\ instance.

Now consider the uniform case, i.e.\ $\uset = E$. 
Let $(x,y)$ be a fractional solution to the corresponding (\lp), and
select some edge $f\in\uset$.  
Since $x^{-f}$ is contained in 
$P_G \cap \{x\in \R^{E} \,\mid\, x_f =0\}$, 
there exist positive scalars $\lambda^{-f}_1, \cdots,\lambda^{-f}_k$ with
$\sum_{i\in [k]} \lambda^{-f}_i = 1$, and perfect 
matchings $M^{-f}_1, \cdots, M^{-f}_k$ in $G-\{f\}$
such that  
$x^{-f} = \sum_{i\in [k]} \lambda^{-f}_i \chi^{M^{-f}_i}$.
By Caratheodory's theorem, there is a decomposition of the 
latter type with $k$ bounded by $O(m) = O(n^2)$. 
Furthermore given $x^{-f}$, such a decomposition can be computed in
polynomial time using polyhedral techniques~\cite[Thm. 6.5.11]{groetschel_et_al_book}.
\begin{algorithm}
  \caption{: Randomized $O(\log n)$-Approximation for \rap}
	\begin{algorithmic}[1]	
		\label{alg:randomized_rounding}
		\REQUIRE{ $G=(R\dcup T,E)$ with $|R|=|T|$, and $c\in\R^E_{\geq 0}$.}
		\ENSURE{ A feasible solution $\rset$  to \rap\, on $G$ with $\uset=E$
			and cost vector $c$.}
		\STATE{ Solve (\lp) to obtain an optimal solution $(x,y)$}
		\STATE{ $\rset \gets \emptyset$ }
		\WHILE{ $\rset$ is infeasible }
		\STATE{ Select an edge $f\in \uset$ such that $\rset\setminus \{f\}$ contains no
			perfect matching }
		\STATE{ Compute a decomposition of $x^{-f}$ as $x^{-f}= \sum_{i=1}^k
			\lambda_i^{-f} \chi^{M^{-f}_i}$ and select one matching $\bar M \in\{M^{-f}_i\mid i \in [k]\}$ with
			$\mathrm{Pr}\left[\bar M = M^{-f}_i\right] = \lambda^{-f}_i$ for all $i \in [k]$}
		\STATE{ Add to $\rset$ all edges from $\bar M$ that connect distinct
			connected components in $(R \dcup T,\rset)$}
		\ENDWHILE
		\RETURN{$\rset$}
	\end{algorithmic}
\end{algorithm}

The algorithm performs several iterations of randomized rounding that
are based on the latter decomposition of fractional matchings.
More precisely, at each iteration, an infeasible set $\rset\subseteq E$ of edges, that was
chosen so far, is augmented with an additional set $M$ of edges chosen randomly
as follows. First, an arbitrary edge $f$ is chosen from $E$ 
among all edges not yet covered by $\rset$, i.e.\ among all $e^\prime\in E$ such that the edge set $\rset$ selected so far contain
no perfect matching that does not include $e^\prime$. 
Next, a decomposition of the vector $x^{-f}$ as a convex combination of perfect matchings is computed, as above. This decomposition is then used 
to select a single perfect matching $\bar M$ from $\{M^{-f}_1, \cdots, M^{-f}_k\}$ \emph{randomly}, where 
$M^{-f}_i$ is chosen with probability $\lambda^{-f}_i$ for all $i\in [k]$. Finally, the augmenting set 
$M\subseteq \bar M$ is chosen to contain all edges of $\bar M$ connecting 
\emph{distinct connected components} of $\rset$. The edges of $M$ are added to $\rset$ and 
the algorithm proceeds to the next iteration. 
The algorithm terminates when $\rset$ is a feasible solution. 
A summary of the algorithm is presented as Algorithm~\ref{alg:randomized_rounding}.

The correctness of the algorithm is shown by exploiting the following
classical result for matching-covered graphs.
\newpage
\begin{theorem}[{\cite[Thm. 4.1.1., p. 122]{lovasz_plummer_book_86}}]
\label{thm:m-c-graphs-equivalence}
A connected bipartite graph $H=(U\cup W, E)$ with $|U\cup W| \geq 4$ is matching-covered 
if and only if  for any $u \in U$ and $w \in W$ the graph $H - \{u\} -
\{w\}$ has a perfect matching. 
\end{theorem}
Concretely, Theorem~\ref{thm:m-c-graphs-equivalence} serves as a main
ingredient  
for the proof of the next lemma, which states a useful structural
property of intermediate solutions in the algorithm.  
\begin{lemma}\label{lem:randomized-rounding-algo-feasibility}
 Let $\rset$ be a non-empty set of edges already selected in an 
 arbitrary iteration of Algorithm~\ref{alg:randomized_rounding}.
 Then, the subgraph $G[\rset]$ induced by $\rset$ is
 matching-covered.    
\end{lemma}
\begin{proof}
As $\rset$ is assumed to be non-empty, $\rset$ contains at least one perfect
matching of $G$. Thus, $G[\rset]$ spans $G$, i.e.\ $G[\rset]=(R\dcup
T,X)$, and does not have isolated nodes. 
Let $S \subseteq R \dcup T$ be the nodes of some connected
component of $G[\rset]$. It suffices to prove that the corresponding connected component $(S, \rset[S])$
with 
$\rset[S]:=\{e\in \rset\mid e=\{s_1,s_2\}, \textrm{ for some }
s_1,s_2\in S\}
$
is matching-covered. \\
For $|S| = 2$, $\rset[S]$ contains exactly one edge that belongs to a
perfect matching in $\rset$. Thus, the claim is proved. \\
Next, assume that $|S| > 2$. To prove that  $(S,\rset[S])$ is
matching-covered, an induction on the
number of iterations in Algorithm~\ref{alg:randomized_rounding} is performed.
In the first iteration, a perfect matching is added to $\rset$ in Step
6, thus the claim holds in that case. \\
Now, let $\rset^\prime\subseteq \rset$ be the set of edges selected
until the beginning of the iteration preceding the current
iteration. By the induction 
hypothesis, it can be assumed that every connected component of
$(R \dcup T,\rset^\prime)$ is matching-covered. 
To prove the claim, it must be shown that every edge $e\in \rset[S]$ is
contained in some perfect matching of $S$. If $e\in \rset^\prime$, the claim
holds by the inductive hypothesis, and due to $\rset^\prime \subseteq \rset$.
If $e\not\in \rset^\prime$, then $e\in \bar M$, where $\bar M$ is the matching selected in Step 5 in the current iteration. 
This means that $e$ connects nodes from two distinct connected 
components of $(R \dcup T, \rset^\prime)$ since $e$ was added in Step 6.\\
Now, pick any cycle $C\subseteq \rset$ in $G[\rset]$ containing $e$ with a
minimum number of edges from $\bar M$. Let $D_1, \cdots, D_l \subseteq
R \dcup T$ be the components in $(R \dcup T,\rset^\prime)$ that have
edges in $C$. From minimality of $|C\cap \bar M|$ it follows that $C$ is a simple
cycle (i.e.\ each node is contained in at most two of its edges) and
that each component $D_j$, $j=1,\cdots,l$ contains exactly two nodes
incident to the cycle. This cycle can now be used to demonstrate the  
existence of the desired perfect matching $M^\prime$ as follows. First,  
include in $M^\prime$ all edges in $C\cap \bar M$.   
Then, in every component $D_j$ for $j=1,\cdots, l$ pick a matching
covering all nodes, except the two nodes incident to the cycle
$C$. Due to Theorem~\ref{thm:m-c-graphs-equivalence}, such a matching exists
since each component $D_j$ is matching-covered. The matching
chosen so far covers exactly the nodes in $D_1\cup \dots \cup D_l$.
Finally, pick any matching covering all other components of $(R\dcup T,\rset^\prime)$ 
that are not incident to $C$. This matching exists, since
again, $(R\dcup T,\rset^\prime)$ is matching-covered. The  
result is a perfect matching in $G[\rset]$ containing $e$, which
completes the proof.
\end{proof}
  
Lemma~\ref{lem:randomized-rounding-algo-feasibility} guarantees that at every stage in the algorithm, the
only edges not yet covered by the current solution $\rset$ are the isolated edges of
$G[\rset]$. Now, since at an iteration where an uncovered edge $f$ is
chosen in Step~4, the set $M$ must contain two edges distinct from $f$, that
are incident to the endpoints of $f$, the edge $f$ is guaranteed to be covered
in the end of this iteration. This immediately implies that the
algorithm terminates with a feasible solution after at most $|E|$ iterations. 
It hence remains to bound the expected cost of 
the solution returned by Algorithm~\ref{alg:randomized_rounding}.

\begin{lemma}
\label{lem:cost-analysis-rand-rounding}
  The expected cost of the solution returned by Algorithm~\ref{alg:randomized_rounding} is 
  $O(\log n) \cdot \opt$, where $\opt$ is the optimal solution value for the \rap\ instance.
\end{lemma}
\begin{proof}
The feasibility of the obtained solution and the bound on the running 
time are guaranteed by Lemma~\ref{lem:randomized-rounding-algo-feasibility}.
\\
For a set $Q\subseteq E$ of edges, let $c_{\textrm{LP}}(Q)$ denote the contribution of the edges in $Q$ 
to the \lp\ cost, i.e.\ $c_{\textrm{LP}}(Q) = \sum_{e\in Q} c_e y_e$. For a
node $v\in R\dcup T$,  $\delta(v) \subseteq E$ represents the set of edges incident to $v$.
To bound the approximation guarantee, the expectation of the ratio
$c(\rset)/c_{\textrm{LP}}(E)$ is bounded accordingly. 
Since the \lp\ is a relaxation of the problem,  $c_{\textrm{LP}}(E)
\leq \opt$ holds. Thus, this ratio is a valid
bound on the approximation guarantee.
\\
To obtain the desired bound, a scheme charging every selected edge in any
stage of the algorithm to one of its endpoints is developed. It is then shown that the expected cost charged to any
node $v\in V$ is bounded by $O(\log n)$ times the \emph{fractional cost} $c_{\textrm{LP}}(\delta(v))$ associated 
with the node. This implies that the expected cost of all edges added by the algorithm is at most
$$
O(\log n) \cdot \sum_{v\in R\dcup T} c_{\textrm{LP}}(\delta(v)) \leq
O(\log n) \cdot \opt,
$$
where the last inequality follows from linearity of expectation, 
$c_{\textrm{LP}}(E) \leq \opt$ and $c_{\textrm{LP}}(E) = \frac{1}{2}\sum_{v\in R\dcup T} c_{\textrm{LP}}(\delta(v))$.
\\
Next, it is described how the costs of the selected edges are charged to the nodes of the graph.
Let $\bar\rset\subseteq E$ be the set of edges returned by the algorithm.  
Formally, with each node $v\in R\dcup T$, a collection of edges $\bar X_v
\subseteq \bar\rset$ is associated such that $\bigcup_{v\in R\dcup T}
\bar X_v = \bar\rset$ holds and such that
$c(\bar X_v)$ is bounded by $O(\log n)$ times
the fractional load at $v$ in expectation. 
\\
The sets $\bar X_v$ are constructed as follows. In the beginning $\bar
X_v = \emptyset$ for all $v\in R\dcup T$.
Let $\rset$ be the set of edges selected so far by the algorithm
and let $M\subseteq E\setminus \rset$ be the set of edges selected to
be added to $\rset$ in Step 6 of the current iteration. At this stage,
the sets $\bar X_v$ might already contain some edges. 
Depending on the selection of $M$, the sets $\bar X_v$ now change as follows.
Consider an edge $f = \{r,t\} \in M$. 
Recall that the algorithm only includes edges in the solution if they connect
different connected components in $(R\dcup T,\rset)$. Thus, $r$ and $t$ lie in different connected components.
Let $D_r$ and $D_t$ be the node sets of the connected components to which $r$ and $t$
belong, respectively, and assume without loss of generality that $|D_r| \leq |D_t|$.
Then, $f$ is charged to $r$, i.e.\ $f$ is included in $\bar X_r$. In other words, an edge 
added by the algorithm in any iteration is charged to the node contained in the \emph{smaller
	connected component}, with ties broken arbitrarily.
\\
It is obvious that the latter scheme charges all edges in $\bar\rset$ to some nodes, such 
that $\bigcup_{v\in V} \bar X_v = \bar\rset$ holds in the end of the
last iteration. It remains to analyze the quantity $c(\bar X_v)$
for a single node $v\in R\dcup T$. The bound on $c(\bar\rset)$ will then follow from linearity of expectation
and the previous discussion. To arrive at the desired bound it suffices to make the following 
two observations. 
\\
First, at any time, if an edge is charged to $v$, its expected 
cost is at most $c_{\textrm{LP}}(\delta(v))$. Indeed, recap that the edges in $M$ come from 
a perfect matching chosen at random from the decomposition of some fractional perfect matching $x^{-f}$
in the graph (this $x^{-f}$ corresponds to the edge $f$ chosen in Step
4 in the current iteration). 
Let this decomposition be
$$
x^{-f} = \sum_{i\in [k]} \lambda^{-f}_i \chi^{M^{-f}_i}.
$$
The distribution over the integral matchings defining 
$x^{-f}$ induces a distribution over the edges incident to $v$: each edge $e\in \delta(v)$
is contained in the perfect matching with a probability $p_{e} \in [0,1]$ given by
$$
p_{e} = \sum_{i\in [k] \,:\, e\in M_i^{-f}} \lambda^{-f}_i =
x^{-f}_{e}.
$$
Since $x^{-f}_{e} \leq y_{e}$ for all $e\in E$, the expected cost of the 
edge charged to $v$ is at most $\sum_{e\in \delta(v)} c_{e}
x^{-f}_{e} \leq c_{\textrm{LP}}(\delta(v))$, 
proving the first property.
\\
The second observation concerns the number of times the node $v$ is charged in the
course of the algorithm. Consider any iteration in which some edge was charged to
$v$, and let $D_v$ be the nodes in the component of $v$ in the beginning of the
iteration. Since an edge is only charged to a node of the smaller component,
and since charged edges always merge connected components, the size of the connected
component containing $v$ in the end of the iteration is \emph{at least} $2|D_v|$.
Since the graph only contains $n$ nodes, this doubling can only happen at most
$\log n$ times.
\\
Therefore, $c(\bar X_v)$ is, in expectation, indeed at most
$O(\log n)\cdot c_{LP}(\delta(v))$,
which concludes the proof.
\end{proof}

The proof of Lemma~\ref{lem:cost-analysis-rand-rounding} uses a discharging argument that
assigns the cost of selected to nodes in the graph, in a way that the total 
assigned cost to any node can be bounded by the fractional contribution of edges
incident to the node. While it seems difficult to perform this type of argument
a posteriori for a computed feasible solution, a cost assignment 
scheme is designed, that is constructed by following the progress of the algorithm.

Lemma~\ref{lem:randomized-rounding-algo-feasibility} and
Lemma~\ref{lem:cost-analysis-rand-rounding} already imply the correctness of
Theorem~\ref{thm:RAP-admits-log-n-approximation-algorithm} for the uniform case. 
The generalization to the non-uniform case is explained in the following proof.
\begin{proof}[Proof of Theorem~\ref{thm:RAP-admits-log-n-approximation-algorithm}]
It remains to show how to treat the case $\uset \neq E$. 
For this, a transformation to reduce such an instance to a uniform instance by losing only a factor of 
$2$ in the approximation guarantee is provided. 
\\
The transformation first adds to the graph one parallel edge $\bar e$ for 
every $e\not\in \uset$. Let $G^\prime$ be the obtained graph with
edge set $E^\prime$. The new set of vulnerable edges is set to $\uset^\prime = E^\prime$. Solutions 
for the two \rap\ instances are in the following correspondence: A solution 
$\rset \subseteq E$ to the original instance on $G$ can be transformed
to a solution for the new instance on $G^\prime$ of at most double the cost by taking 
$\rset^\prime = \rset\cup \{\bar e \,\mid\, e\in \rset\setminus \uset\}$. Conversely,
a solution $\rset^\prime$ for the new instance can be transformed to a solution
for the original instance with the same, or better cost, by choosing
$\rset = \rset^\prime \cap E$. 
\\
Let $\opt^\prime$ denote the optimal solution value of the transformed uniform 
instance. Our $O(\log n)$-approximation algorithm for the general case proceeds
by first transforming the instance to a uniform instance of \rap, as above,
then invoking Algorithm~\ref{alg:randomized_rounding} to obtain the set $\rset^\prime$,
having expected cost at most $O(\log n) \opt^\prime =
O(\log n) \opt$ 
and then returning $\rset= \rset^\prime \cap E$.
\end{proof}

The section is concluded by arguing why simpler randomized 
rounding techniques, for instance, ones that lead to logarithmic 
approximation to many covering problems, are unlikely to lead to 
a similar result for \rap. The reason for this is that there does
not seem to be a simple way to obtain a compact set cover-type 
representation of \rap\ without losing a super-logarithmic factor
in the approximation guarantee. One natural attempt could be to
consider every vulnerable edge $f\in\uset$ as an \emph{element}
that needs to be covered, and every possible perfect matching 
$M\subseteq E$ that does not contain $f$, a \emph{covering set}
that covers the edge $f$ (and all other edges in $\uset\setminus M$).
The cost of the covering set is simply the sum of the costs of edges
in the corresponding perfect matching.
Unfortunately, it is easy to come up with examples in which
the optimal solution value in the latter set covering model has 
cost $\Omega(n) \opt$. Such instances can be constructed, 
for example by choosing an instance, such that any feasible 
solution must have some nodes with very high degree, while an
optimal solution has cost $O(n)$.

\section{Hardness and approximability of card-RAP}
\label{sec:card-rap-hardness-and-approximation}

In this section we restrict our attention to the special case
\urap, i.e.\ to the problem of finding a robust assignment of minimum cardinality.
Subsection~\ref{sec:cardrap-is-hard-to-approximate} deals with the proof of Theorem~\ref{thm:card-rap-is-hard-to-approximate}
stating that \urap\ is hard to approximate up to some constant $\delta >1$.
This implies thtat \urap\ does not admit a {\small PTAS} unless $\p = \np$.
An $O(1)$-approximation algorithm for \urap\ is then presented
in Subsection~\ref{sec:algorithm_unweighted_RAP}, thus proving
Theorem~\ref{thm:approximation_unweighted_RAP}.

\subsection{Hardness of Approximation for card-RAP}
\label{sec:cardrap-is-hard-to-approximate}

As is the case for general uniform \rap, the hardness proof we present here
for \urap\ is also based on a reduction from
the Set Cover Problem (\setcover). 

By Lemma~\ref{lem:set-cover-rap-feasibility-uniformversion},
any given \setcover\ instance $\inst:=([k],\mS)$ can be equivalently
transformed, in time polynomial in the input size of $\inst$, into a uniform \rap\ instance
$\bar \inst:=(\bar G, \bar E, \bar c)$. 
Note that $\bar G:=(\bar R\dcup \bar T,\bar E)$ is the graph constructed from $\inst$ by applying the
transformation steps {\small\bf (T1)--(T4)} as stated on page
\pageref{tsteps-one-to-four} and step {\small\bf (T5)} from page
\pageref{tstep-five}, whereas the cost vector $\bar c\in\R^{\bar E}_{\geq
  0}$ is chosen with $\bar c_e=1$, if $e$ belongs to $\bar E_4=E_4$, and $\bar c_e=0$, otherwise.
Recap further that the specific choice of $\bar c$ is
crucial for deriving an equivalent reformulation of $\inst$ in terms
of a uniform \rap. 
Lemma~\ref{lem:set-cover-rap-feasibility-uniformversion}~(b) indeed
shows  that all edges from $\bar G$ except those from $E_4$ can be assumed to be contained in an optimal solution
to $\bar\inst$ as they have zero cost. 
This property does not hold anymore when $\bar\inst$ is replaced with its unweighted version.

In order to arrive at a \urap\ instance being equivalent to the
given \setcover\ instance $\inst$, a further transformation step to be
performed on $\bar G$ is needed.
\begin{itemize}
\label{tstep-six}
\item[{\small\bf (T6)}] 
  For each edge $f_s = \{\bar u_s, u_s\} \in E_1$ subdivide the edge by adding two new nodes 
  $z^s_1$ and $z^s_2$.
  \begin{figure}[htb]
    \centering
    \begin{tikzpicture}[inner sep = 2pt]
      
      \node (0) at (0, 0) [circle, draw] {$\bar{u}_s$};
      \node (1) at (4.5, 0) [circle, draw] {$u_s$};
      \node (2) at (1.5, 0) [circle, draw] {$z^s_1$};
      \node (3) at (3, 0) [circle, draw] {$z^s_2$};
      
      \draw[-, thick] (0) -- (2) -- (3) -- (1);
      
      \node at (-1,0) {$\mapsto$};
      \node (6) at (-3,0) [circle,draw] {$\bar{u}_s$};
      \node (7) at (-2,0) [circle,draw] {$u_s$};
      \draw[-, thick] (6) -- (7);
    \end{tikzpicture}
    \caption{Replacing an edge $f_s = \{\bar u_s, u_s\} \in E_1$ by a path of
      length three.}
    \label{fig:subdividing-the-edges-for-the-card-rap-hardness-proofs}
  \end{figure}
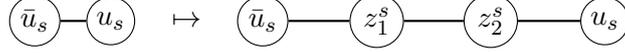  
\end{itemize} 
Transformation step {\small\bf (T6)} is illustrated in
Figure~\ref{fig:subdividing-the-edges-for-the-card-rap-hardness-proofs}.

In the following, let $\tilde{G}=(\tilde{R} \dcup \tilde T,\tilde E)$,
with $\tilde E=\bigcup_{j=1}^6\tilde E_j$, denote the graph obtained by applying transformation
step {\small\bf (T6)} to $\bar G$, and let $\tilde{\inst}$ be the
corresponding uniform \urap\ instance defined on $\tilde{G}$, i.e.\
$\tilde\inst=(\tilde{G},\tilde{E},\mathbf{1}_{\tilde E})$ where
$\mathbf{1}_E$ is the all one vector in $\R^{\tilde E}$. 
Observe that 
\begin{itemize}
\item $\tilde{R}=\bar R\cup\{z_2^s\mid s\in [k]\}$,  $\tilde{T}=\bar
  T\cup\{z_1^s\mid s\in [k]\}$, 
\item $\tilde E_1:=\big\{\, \{\bar u_s, z_1^s\},
  \{z_1^s,z_2^s\},\{z^s_2,u_s\}\mid s\in[k]\big\}$,
\item $\tilde E_3 = \bar E_3$,  $\tilde E_5=\bar E_5$,
\item
  $\tilde E_2 = \bar E_2 = E_2$,  $\tilde E_4 = \bar E_4 = E_4$, and $\tilde E_6 = \bar E_6 = E_6$
\end{itemize}
hold. In order to formulate an analogue of
Lemma~\ref{lem:set-cover-rap-feasibility-version-1} the following notion of \emph{efficiency} for feasible
solutions to $\tilde{\inst}$ is useful.

Concretely, a feasible solution $\tilde\rset \subseteq \tilde E$ to
$\tilde{\inst}$ is called \emph{efficient}, if $|\tilde\rset \cap \tilde E_2| = |\tilde \rset \cap \tilde E_6| = k$. 
To justify this definition, consider any non-efficient feasible
solution $\rset$. On the one hand, feasibility of $\rset$ implies that
$|\rset \cap \tilde E_2| \geq k$, as every node $u_i$, $i\in [k]$,
must have at least one incident edge in $\tilde E_2$, that is
contained in $\rset$. Similarly, $|\rset \cap \tilde E_6|\geq k$ must hold. On the other
hand, it is next argued that $\rset$ can be transformed in polynomial
time into a new feasible solution $\tilde \rset$ that is efficient. 

By a similar argument as used in the proof of Lemma~\ref{lem:set-cover-rap-feasibility-version-1},
it can be shown that $\mathcal{C}_{\rset} = \{S\in \mS \,\mid\,
\{\bar{v}_S,\tilde{v}_S\}\in \rset\}$ is a feasible cover to the given
\setcover\ instance $\inst$. 
Now, for each $i\in [k]$, choose an arbitrary set $S_{j_i} \in
\mathcal{C}_{\rset}$ with $i\in S_{j_i}$. Then,
\[ \tilde \rset := \big(\,\rset \setminus (\tilde E_2 \cup \tilde E_6)\big)\; \cup\; 
\big\{\,\{u_i, v_{S_{j_i}}\}, \{\bar u_i, w_{S_{j_i}}\} \mid  i\in [k]\big\} \]
is a feasible to $\tilde\inst$ with $|\tilde \rset| \leq |\rset|$. 
In addition, $\tilde X$ is efficient and can be constructed in polynomial time.

\begin{lemma}
\label{lem:set-cover-rap-feasibility-version-2}
Let $\inst=([k],\mS)$ be any \setcover\ instance, and let $\tilde
\inst$ be the corresponding uniform \urap\ instance on
the graph $\bar{G}:=(\tilde R\dcup \tilde T,\tilde E)$ that is obtained by applying the
transformation steps {\small\bf (T1)--(T6)} on $\inst$.
Moreover, let $q:= | \tilde E_1 \cup \tilde E_3 \cup \tilde E_5 |$.
\begin{itemize}
\item[(a)]
  $\tilde{\inst}$ can be constructed from $\inst$, in polynomial time.
\item[(b)]
  Let $\tilde \rset \subseteq \tilde E$ be an efficient feasible
  solution to $\tilde\inst$.
  Then there exists a feasible cover for $\inst$, that is of 
  size $|\tilde \rset| - q - 2k$. 
 \item[(c)] 
   If $\mathcal{C}$ is a feasible cover for $\inst$, then there exists
   a feasible solution $\tilde X$ to $\tilde\inst$ with  $|\tilde
   X|=|\mathcal{C}| + q + 2k$. 
\end{itemize}
\end{lemma}
\begin{proof}
	\text{}
  \begin{itemize}
  \item[(a)]
    This statement follows from Lemma~\ref{lem:set-cover-rap-feasibility-uniformversion}~(a) and the fact that
    performing transformation step {\small\bf (T6)} on $\bar{G}$ can also
    be done in polynomial time.
  \item[(b)]
    Let $\tilde \rset \subseteq \tilde{E}$ be an efficient, feasible
    solution.
    Since every edge from $\tilde G$ is vulnerable and since every
    edge in $\tilde E_1\cup \tilde E_3 \cup \tilde E_5$ is incident to
    some node of $\tilde G$ with degree two, it follows that
    $\tilde E_1 \cup \tilde E_3 \cup \tilde
    E_5 \subseteq \tilde \rset$  must hold.
    Furthermore, $|(\tilde E_2 \cup \tilde E_4) \cap \tilde \rset| =
    2k$ implies that $\tilde \rset$ contains $|\tilde \rset| - q - 2k$ edges from
    $\tilde E_4$.
    Thus, $\mathcal{C}_{\tilde\rset} = \big\{S\in \mS
    \,\mid\, \{\bar{v}_S,\tilde{v}_S\} \in \tilde \rset \big\}$ is a
    feasible cover for $\inst$ with $\mathcal{C}_{\tilde\rset}=|\tilde \rset| - q - 2k$.
\item[(c)]
  Let $\mathcal{C}$ be a feasible cover for $\inst$. For each
  $i\in[k]$, let further denote by $j_i$  the index of some set of $\mS$
  that covers $i$ in $\mathcal{C}$. Then, 
  \begin{align*}
    \tilde \rset := \tilde E_1 & \cup \tilde E_3 \cup \tilde E_5 \\
    & \cup \{\{u_i, v_{S_{j_i}}\}, \{\bar u_i, w_{S_{j_i}}\} \,\mid\,  i\in [k]\}
\cup \{\{\bar v_{S_j}, \tilde v_{S_j}\} \,\mid\, S_j\in
\mathcal{C}\} 
\end{align*}
is a feasible solution to $\tilde \inst$, that additionally satisfies $|\tilde \rset| =
|\mathcal{C}| +q+2k$.
\end{itemize}
\end{proof}

Now, Theorem~\ref{thm:card-rap-is-hard-to-approximate} can be proved with the help of Lemma~\ref{lem:set-cover-rap-feasibility-version-2}.

\begin{proof}[Proof of Theorem~\ref{thm:card-rap-is-hard-to-approximate}]
  It is well known that the \emph{Vertex Cover Problem in sub-cubic
    Graphs} (\vcthree) on an input graph $H=(V,E_H)$ with $|E_H|=k$ can be equivalently restated as an instance of the Set Cover Problem, where the ground set is $E_H\cong[k]$ and each $S\in\mS$ corresponds to a cut set $\delta(v)$, for some $v\in V$.  Note that $|S|\leq 3$, for all $S\in\mS$, and $|\{S\in \mS \,\mid\, s\in S\}| = 2$ for all $s\in [k]$.
  Alimonti and Kann proved in~\cite{alimonti_viggo_97} that there exists a constant $\delta >1$ 
  such that \vcthree\ does not admit a polynomial $\delta$-approximation algorithm unless \p$=$\np.
  \\
  Now, let a \vcthree\ instance be given in terms of a \setcover\ instance
  $([k],\mS)$.  Let $\tilde G$ be the graph obtained form  $([k], \mS)$ by applying
  the transformation steps {\small\bf  (T1)--(T6)}, let $\tilde{\inst}$ be
  the uniform \urap\ instance induced by $\tilde G$.
  \\
  Observe that $\tilde G$ has $O(k)$ edges and that any feasible solution to 
  \vcthree\ contains at least $k/3 = \Omega(k)$ sets.
  From Lemma~\ref{lem:set-cover-rap-feasibility-version-2}, it follows
  that for any constant $\alpha > 1$ 
  there exists a constant $\beta = \beta(\alpha) >1$ such that any 
  polynomial $\beta$-approximation algorithm for the uniform \urap\
  can be used to construct a polynomial $\alpha$-approximation algorithm for \vcthree. Since 
  no $\delta$-approximation exists for \vcthree, unless \p$=$\np, there is also
  no $\alpha(\delta)$-approximation for the uniform \urap\, unless
  \p$=$\np. 
\end{proof}

\subsection{$O(1)$-Approximation for card-RAP} 
\label{sec:algorithm_unweighted_RAP}  

In this subsection, a constant factor approximation algorithm for
\urap\ is developed.
The algorithm makes use of an ear decomposition of the underlying bipartite graph. 
At the end of this subsection it is also shown that the algorithm presented can be easily extended to unbalanced graphs.\\
Note that similar approaches using ear decompositions have successfully
been used to approximate various combinatorial optimization problems.
Among others, these works include algorithms for the minimum edge 
connected subgraph problem by Cheriyan, Seb{\H{o}} and
Szegeti~\cite{cheriyan_et_al_01} as well as the traveling salesman problem
by Seb{\H{o}} and Vygen~\cite{sebHo2014shorter}.

\begin{definition}[Ear Decomposition of Bipartite Graphs (e.g.,
  p.~123 in \cite{lovasz_plummer_book_86})]\label{def:bipartite-ear-decomposition} 
  \text{}\\
  Let $H$ be a bipartite graph, and let $H^\prime$ be a subgraph of $H$.  
  An \emph{odd ear of $H$ with respect to $H^\prime$} is a path $P$ in $H$ with an odd number of edges
  and such that $P$ and $H^\prime$ have exactly two nodes in common. Those two nodes form the end points 
  of $P$, and belong to different parts of the bipartition.\\
  A \emph{bipartite ear decomposition of a bipartite graph} $H$ is a sequence $P_0, P_1,\dots, P_q$ of paths in $H$, such that:
	\begin{itemize}
  \item [(i)] $P_0 = (\{v_1,v_2\}, \{\{v_1,v_2\}\})$ is a graph composed of a single edge,
  \item [(ii)] $H = P_0 + \dots + P_q$,
  \item [(iii)] for every $j=1,\dots, q$, the path $P_j$ is an odd ear with respect to 
    $H_{j-1} := P_0+\dots+P_{j-1}$.
	\end{itemize}
\end{definition}

The next theorem provides a well-known connection between matching-covered bipartite
graphs and bipartite ear decompositions.

\begin{theorem}[{\cite[Thm. 4.1.6]{lovasz_plummer_book_86}}]
\label{thm:bipartite-m-c-iff-has-a-ear-decomposition}
A bipartite graph is matching-covered if and only if it has a bipartite ear decomposition.
\end{theorem}

Now consider an instance $\inst = (G,\uset)$ of \urap\ with uncertainty set $\uset \subseteq E[G]$.
As always, we assume that $\inst$ is feasible.
In order to apply the theory of matching-covered graphs the algorithm first removes from $G$ all so-called
\emph{dispensable edges}, i.e.\ all edges not appearing in any perfect matching of $G$. 
This way, a graph is obtained that is, by definition, matching-covered. 
Note that the removal can be implemented in polynomial time using any efficient algorithm for finding bipartite matchings.
We also assume that the removal of the dispensable edges results in a connected graph (otherwise we treat each connected component separately).
Moreover, recap that dispensable edges can be always removed from any feasible solution without breaking feasibility.
In the following, the new graph resulting from deleting all dispensible edges is also called $G$ to facilitate readability. 

Now let $G = P_0 + \dots + P_q$ be any ear decomposition of $G$
(which is by no means unique) with some arbitrary chosen initial edge $P_0$. 
An ear $P_j$ of the decomposition is called \emph{trivial} if it is not $P_0$ and if it consists of a single edge, only.
The next lemma shows that a feasible solution to $\inst$ can be obtained from the ear decomposition of $G$ by skipping trivial ears.

\begin{lemma}
\label{lem:trivial-ears-removal-maintains-m-c-property}
Let $J := \{j\in [q] \,\mid\, P_j \,\,\, \text{is a trivial ear}\}$, and let $G^\prime$ be the subgraph in $G$ resulting from $G$ by removing all
trivial ears, i.e. $G^\prime = P_0 + \sum_{i\in [q]\setminus J} P_i$. 
Then, $\rset := E[G^\prime]$ is a feasible solution to the \urap\ instance
on the original graph $G=(R\dcup T,E)$. 
Furthermore, $|\rset| \leq 3|T|$ holds.
\end{lemma}

\begin{proof}
The subgraph $G^\prime$ is obtained from $G$ by deleting trivial ears only, i.e.\ $G^\prime$ is spanning in $G$.
By Proposition~\ref{prop:matching_covered_and_RAP}, it suffices to show that $G^\prime$ is matching-covered and that there are no isolated edges.
The fact that $G^\prime$ is matching-covered follows from the fact that it has, by definition, an ear decomposition $G^\prime = P_0 + \sum_{i\in [q]\setminus J} P_i$ and from Theorem~\ref{thm:bipartite-m-c-iff-has-a-ear-decomposition}.
Since $G$ has a feasible solution to \urap, there are no isolated vulnerable edges in $G$, and hence also not in $G^\prime$. 
This proves that $\rset = E^\prime$ is feasible.\\

It remains to bound the number of edges in $\rset = E[G^\prime]$.
Let $G^\prime = Q_0 + \ldots + Q_p$ be the ear decomposition of
$G^\prime$ consisting of the non-trivial ears of $G$ (appearing
in the same order) with $p := q - |J|$.
Furthermore, define $l_j := |T \cap (V[Q_j] \setminus (V[Q_0] \cup
\dots \cup V[Q_{j-1}]))|$ as the number of internal task nodes in the ear $Q_j$ for $j=1,\dots,p$. 
Since the ears $Q_j$ are non-trivial we have $l_j \geq 1$, $p \leq
|T| - 1$ and $\sum_{j=0}^p l_j = |T|$ implying 
\[ |\rset| \leq 1 + \sum_{j=1}^p 2 (l_j + 1) = 1 + 2 |T| + p \leq 3 |T| . \] 
\end{proof}

Lemma~\ref{lem:trivial-ears-removal-maintains-m-c-property} allows to
arrive at an approximation algorithm summarized as Algorithm~\ref{alg:ear-decomp-approx-algo-bipartite-uniform-case}.
The analysis is deferred to the proof of Theorem~\ref{thm:approximation_unweighted_RAP}. 

\begin{algorithm}
\caption{: $O(1)$-Approximation for \urap}
\begin{algorithmic}[1]	
\label{alg:ear-decomp-approx-algo-bipartite-uniform-case}
\REQUIRE{$G = (R\dcup T,E)$ and $\uset \subseteq E$.}
\ENSURE{ a feasible solution $\rset$ to the \urap\ instance on $G$ and $\uset$}
\STATE{ $\rset \gets \emptyset$ }
\STATE{ Remove all dispensable edges form $G$}
\STATE{ Compute an ear decomposition $G = P_0 + \ldots + P_q$}
\STATE{ $\rset \gets P_0\; \cup\; \bigcup\{E[P_j] \,\mid\, P_j \,\,\text{is not trivial}, ~j=1, \ldots, q\}$}
\RETURN{$\rset$}
\end{algorithmic}
\end{algorithm}

\begin{proof}[Proof of Theorem~\ref{thm:approximation_unweighted_RAP}]
According to de Carvalho and Cheriyan~\cite{carvalho_cheriyan_05}, an
ear decomposition of matching-covered graphs can be computed in polynomial time. 
Furthermore, all other computations can also
be implemented efficiently, such that the running time of 
Algorithm~\ref{alg:ear-decomp-approx-algo-bipartite-uniform-case} is polynomial.\\  
From Lemma~\ref{lem:trivial-ears-removal-maintains-m-c-property}, it follows that the set $\rset$ returned 
by Algorithm~\ref{alg:ear-decomp-approx-algo-bipartite-uniform-case} is feasible. 
For $\uset = E$, any feasible solution must have at least two edges incident
to any node from the set $T$.
Hence, $\opt \geq 2 |T|$. Since $|\rset|\leq 3 |T|$, the approximation
guarantee is indeed $1.5$.
If $\uset \subsetneq E$, then $G$ can contain a perfect matching not including any edge from $\uset$, hence $\opt \geq |T|$. 
Using  $\opt \geq |T|$ yields an approximation factor of $3$.
\end{proof}

In the example below, we provide a family of graphs, for which Algorithm~\ref{alg:ear-decomp-approx-algo-bipartite-uniform-case} produces a sequence of solutions whose objective function values can be arbitrarily close to $1.5\,\opt$ in the uniform case.   

\begin{example}
  \label{rem:ear-docomposition-algorithm-cannot-be-better-than-one-point-five}
  For $k\geq 3$, let $G_k$ be a bipartite graph with node set $\{0,1,\hdots,2k+1\}$. The edges of $G_k$ are defined as follows. $G_k$ contains the edge $\{0,1\}$ as well as the paths $P_i$, $i\in\{2z\mid z=1,\hdots,k\}$, from $0$ to $1$ through nodes $i$ and $i+1$. Additionally, $G_k$ contains the cycles $C_j$, $j\in\{2z\mid z=2,\hdots,k-1\}$ with node set $\{j, j+1, j+2,j+3\}$. Figure~\ref{fig:feasible-slns-do-not-contain-optimal-solutions} shows the graph $G_k$, for $k=3$.
  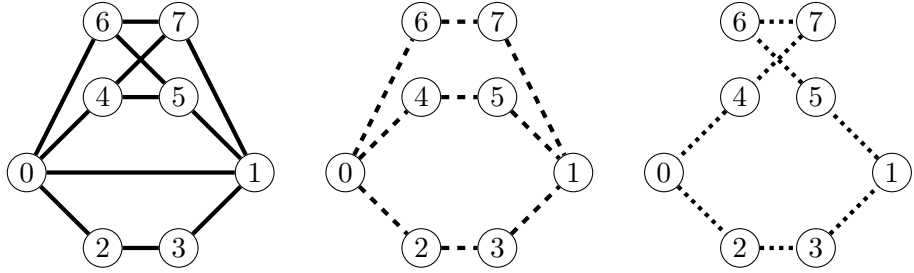
\begin{figure}[H]
    \centering
    \begin{minipage}{0.2\linewidth}
      \begin{tikzpicture}[inner sep = 2pt]
        \node (0) at (0,1) [circle, draw] {$0$};
        \node (2) at (1,3) [circle, draw] {$6$};
        \node (4) at (1,2) [circle, draw] {$4$};
        \node (6) at (1,0) [circle, draw] {$2$};
        \node (3) at (2,3) [circle, draw] {$7$};
        \node (5) at (2,2) [circle, draw] {$5$};
        \node (7) at (2,0) [circle, draw] {$3$};
        \node (1) at (3,1) [circle, draw] {$1$};
        
        \draw[-, ultra thick] (0) -- (2) -- (3) -- (1);
        \draw[-, ultra thick] (0) -- (4) -- (5) -- (1);
        \draw[-, ultra thick] (0) -- (6) -- (7) -- (1);
        \draw[-, ultra thick] (2) -- (5);
        \draw[-, ultra thick] (3) -- (4);
        \draw[-, ultra thick] (0) -- (1);
      \end{tikzpicture}
    \end{minipage}
    \hspace{0.1\linewidth}
    \begin{minipage}{0.2\linewidth}
      \begin{tikzpicture}[inner sep = 2pt]
        \node (0) at (0,1) [circle, draw] {$0$};
        \node (2) at (1,3) [circle, draw] {$6$};
        \node (4) at (1,2) [circle, draw] {$4$};
        \node (6) at (1,0) [circle, draw] {$2$};
        \node (3) at (2,3) [circle, draw] {$7$};
        \node (5) at (2,2) [circle, draw] {$5$};
        \node (7) at (2,0) [circle, draw] {$3$};
        \node (1) at (3,1) [circle, draw] {$1$};
        
        \draw[-, ultra thick, dashed] (0) -- (2) -- (3) -- (1);
        \draw[-, ultra thick, dashed] (0) -- (4) -- (5) -- (1);
        \draw[-, ultra thick, dashed] (0) -- (6) -- (7) -- (1);
\end{tikzpicture}
\end{minipage}
\hspace{0.1\linewidth}
\begin{minipage}{0.2\linewidth}
\begin{tikzpicture}[inner sep = 2pt]
\node (0) at (0,1) [circle, draw] {$0$};
\node (2) at (1,3) [circle, draw] {$6$};
\node (4) at (1,2) [circle, draw] {$4$};
\node (6) at (1,0) [circle, draw] {$2$};
\node (3) at (2,3) [circle, draw] {$7$};
\node (5) at (2,2) [circle, draw] {$5$};
\node (7) at (2,0) [circle, draw] {$3$};
\node (1) at (3,1) [circle, draw] {$1$};

\draw[-, ultra thick, dotted] (0) -- (6) -- (7) -- (1) -- (5) -- (2) -- (3) -- (4) -- (0);
\end{tikzpicture}
\end{minipage}
\caption{Graph $G_3$ (left), a bad solution found by Algorithm~\ref{alg:ear-decomp-approx-algo-bipartite-uniform-case} (center),
  an optimal solution for $\uset=E$ (right).}
\label{fig:feasible-slns-do-not-contain-optimal-solutions}
\end{figure}

In case $\uset=E$, an optimal solution of the corresponding \urap\ instance is given by a Hamiltonian cycle of size $2k+2$, while a ``worst case'' solution that can be found by Algorithm~\ref{alg:ear-decomp-approx-algo-bipartite-uniform-case} is given by the union of all paths $P_i$, i.e.\ by $\bigcup_{i\in\{2z\mid z=1,\hdots,k\}} P_i$. 
This yields
\[ \frac{\alg}{\opt} = \frac{3k}{2(k+1)} \xrightarrow[\infty]{k} \frac{3}{2} . \]
\end{example}

It is worth to remark that Algorithm~\ref{alg:ear-decomp-approx-algo-bipartite-uniform-case} presented above can also be applied to unweighted \rap\ instances defined on unbalanced bipartite graphs, leading to the same approximation guarantees.
The remainder of this section explains the extension of the algorithm.

Consider a bipartite graph $G_u=(R_u\dcup T_u, E_u)$ with
$|T_u| < |R_u|$.
As an initial pre-processing step, the graph $G_u$ is transformed into
a balanced bipartite graph $G_b$ (cf. Proposition~\ref{prop:transformation-non-balanced-to-balanced-instances}).
Therefore a set $D$ of $|R_u| -|T_u| $ dummy task nodes is introduced. 
Further, each such newly introduced node is then connected with all nodes from $R_u$, yielding a new subclass $E_D$ of edges.
The desired balanced bipartite graph is then given by
$G_b:=\left(R_u\dcup (T_u\cup D), E_u\cup E_D\right)$ and defines a \urap\ instance where the uncertainty set is chosen to be $\uset_b= \uset_u\cup E_D$.
As before all dispensable edges has to be removed from $G_b$.

In Lemma~\ref{lem:trivial-ears-removal-maintains-m-c-property} the set $\rset$ is now defined as $\rset := E[G^\prime] \setminus E_D$.
$E[G^\prime]$ is feasible for the instance on $G_b$ and $\rset$ remains feasible for $G_u$ because the edges $E_D$ are incident only to dummy task nodes in $D$ and not to those in $T_u$.
Thus, $\rset \setminus \{ \fset \} $ still contains a matching covering the whole set $T_u$.
The second part of the proof of Lemma~\ref{lem:trivial-ears-removal-maintains-m-c-property} holds if $T$ is replaced by $T_u$.
The same is true for the proof of Theorem~\ref{thm:approximation_unweighted_RAP}.
Hence, the approximation guarantee remains the same for an unbalanced bipartite graph $G_u$.


\section{Hardness of card-RAP with two vulnerable edges}
\label{sec:proof-sec-singleton-two-scenarios-hardness}

In this section the proof of Theorem~\ref{thm:hardness_unweighted_RAP_two_scenarios}, stating that
\urap\ with only two vulnerable edges is \np-hard, is presented.
The proof requires a complexity analysis of the following optimization problem. 
\begin{problem*}[Shortest Nice Path Problem (\snpp)]
\label{prob:SNPP}
\text{}
 \begin{itemize}
\item \underline{Input:} Tuple $(H,s,t)$, where $H=(U\dcup W, E_H)$ is
  a balanced bipartite graph, and $s\in U$, $t\in W$.  
\item 
  \underline{Output:}
  A nice\footnote{Recap from the 
    theory of matching-covered graphs that a subgraph $H^\prime$ of 
    a graph $H$ containing a perfect matching is called \emph{nice} if 
    $(V[H]\setminus V[H^\prime])$ is perfectly matchable in $H$.}
  shortest $s$-$t$-path in $H$, i.e.\ an $s$-$t$-path $P$
  with the smallest possible number of edges such that there is a
  a matching in $H$ covering all nodes in $H$ that are not
  covered by $P$.
\end{itemize}
\end{problem*} 
  
The proof of Theorem~\ref{thm:hardness_unweighted_RAP_two_scenarios}
is based on the fact that \snpp\ is \np-hard.
This result is later shown in Lemma~\ref{thm:SNP_np_hard}.  
We first show how Theorem~\ref{thm:hardness_unweighted_RAP_two_scenarios} follows
from the latter result.
\begin{proof}[Proof of
  Theorem~\ref{thm:hardness_unweighted_RAP_two_scenarios}]
  It suffices to show that \snpp\ can be reduced, in polynomial time, to \urap\ with only two
  vulnerable edges. 
  Then, the \np-hardness of \snpp\ shown in
  Lemma~\ref{thm:SNP_np_hard} completes this proof.\\
  Consider any instance $\inst:=(H,s,t)$ of \snpp\ with graph $H:=(U \dcup W, E_H)$ and with
  terminals $s\in U$ and $t\in W$. Further, set $n := |U| + |W|$.
  To derive an equivalent \urap\ instance, two new nodes $x,y$ are
  introduced as well as the edges  $f_1 := \{s,x\}$, $f_2 := \{x,y\}$
  and $g :=\{y,t\}$. For an illustration, see Figure~\ref{fig_illustrationGprime}
  This leads to a new balanced bipartite graph $G=(T\dcup R,E)$ 
  with $T := U\cup \{y\}$, $R := W \cup \{x\}$ and
  $E:=E_H\cup\{f_1,f_2,g\}$.\\
  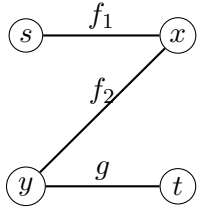
\begin{figure}[hbt] 
    \centering
    \begin{tikzpicture}[inner sep = 2pt]
      \node (0) at (0,2) [circle, draw] {$s$};
      \node (1) at (2,2) [circle, draw] {$x$};
      \node (2) at (0,0) [circle, draw] {$y$};
      \node (3) at (2,0) [circle, draw] {$t$};
      
      \draw[-, thick] (0) to node [above] {$f_1$} (1);
      \draw[-, thick] (1) to node [above] {$f_2$} (2);
      \draw[-, thick] (2) to node [above] {$g$} (3);
    \end{tikzpicture}
    \caption{Additional nodes and edges in the graph $G$.}
    \label{fig_illustrationGprime}
  \end{figure}
  
Now consider the \urap\ instance $\inst^\prime$ defined on $G$ and with
  uncertainty set $\uset=\{f_1, f_2 \}$.
  $\inst^\prime$ can clearly be constructed in polynomial time from
  $\inst$.\\
  Let $P$ be any nice $s$-$t$-path $P$ in $H$. 
  As $s\in U$ and $t\in W$, $P$ consists of an even number of nodes
  and an odd number of edges.
  The definition of nice paths implies that there must exist a
  matching $M_P$ in $H$ covering all nodes $(U\dcup W)\setminus V[P]$,
  i.e.\ all nodes from $H$ not belonging to $P$.\\
  In the graph $G$, the path $P$ and the new edges $f_1,f_2,g$ form the cycle $C:= P
  + (\{x,y\}, \{f_1,f_2,g\})$ with even length. 
  Let $E[C]$ be the set of edges contained in $C$.\\
  It is next argued that the union $M_P\cup E[C]$ is a feasible solution to the \urap\
  instance $\inst^\prime$. 
  Indeed, as the cycle $C$ is even, its edge set $E[C]$ is the union of two perfect
  matchings $M_1$ and $M_2$. This in particulary means that both $M_1$ and $M_2$ cover all nodes from $C$.
  Since the vulnerable edges $f_1,f_2$ are adjacent on $C$, one of these
  matchings contains $f_1$, while the other one contains $f_2$.
  W.l.o.g.\ assume that $f_1\not\in M_1$ and $f_2\not\in M_2$.
  Now recap that $M_P$ covers all nodes from $G$ not contained in $P$.
  Then, for each $f_i$, $i=1,2$, $M_i\cup M_P$ is a perfect matching
  in $G$ not containing $f_i$. This shows that $M_P\cup E[C]$ is
  feasible to $\inst^\prime$.\\
  Summing up, a nice $s$-$t$-path $P$ in $H$ with $L$ nodes (and $L-1$ edges) can
  be used to construct a feasible solution to $\inst^\prime$ with $L +
  2 + \frac{n -L}{2} = \frac{n}{2} + \frac{L}{2} + 2$ edges. 

  On the other hand, it is claimed that the following statements are true.
  \begin{itemize}
    \item[(i)] Any solution $\rset$ feasible to the \urap\ instance
      $\inst^\prime$ must contain a cycle $C$ including $f_1$ and
      $f_2$. 
    \item[(ii)]
      If $\rset$ is \emph{optimal} to $\inst^\prime$ then it must
      addionally hold that
      all edges in $\rset$ not belonging to $C$ form a matching
      covering all remaining nodes.
    \end{itemize} 
  To show the correctness of both statements, observe first that $f_1$ and $f_2$ are contained in any feasible set
  $\rset$ as they are the only edges in $G$ being incident with node $x$. 
  Now, let $M_1, M_2\subseteq \rset$ be two perfect matchings in $G$ 
  with $M_i$ not containing $f_i$, for $i=1,2$. 
  Since $x$ has degree two in $G$ it holds that $f_2\in M_1$ and $f_1\in M_2$.
  Hence,  $M_1\cup M_2 \subseteq\rset$ must contain a cycle involving
  both $f_1$ and $f_2$, implying that statement (i) is true.\\
  If $\rset$ is optimal, it further follows that $\rset = M_1 \cup M_2$ holds.
  Recap that a union of perfect matchings is, in general, a union of even
  cycles with a matching.
  Thus, it still needs to be verified that $X$ cannot contain other
  cycles than $C$.
  For this, assume there is a second cycle 
  $C^\prime\neq C$ in $\rset$.
  Since both $f_1$ and $f_2$ must be contained in the same cycle, the
  existence of $C^\prime$ implies that $C^\prime$ contains neither $f_1$ nor $f_2$. 
  As $C^\prime$ is also even, it contains two different
  matchings covering its nodes. Then, one of those matching can be
  deleted from $\rset$, yielding a proper subset of $\rset$ that is
  still feasible to $\inst^\prime$. This contradicts optimality of
  $\rset$ and shows the correctness of statement (ii).

  Finally, since $y$ has also a degree of two in $G$, edge $g$ has to be part of the 
  cycle $C$ containing $f_1$ and $f_2$. It follows that $C - \{x,y\}$
  is a nice $s$-$t$-path in $H$. 
  Consequently, if for some integer $L$, the cardinality of an optimal solution
  to the \urap\ instance is $\frac{n}{2} + \frac{L}{2} + 2$, the cycle contained
  in this solution has $L + 2$ nodes, and thus it provides a nice $s$-$t$-path with $L$ nodes.

  To conclude, $H$ admits a nice $s$-$t$-path with $L$ nodes if and only if
  the optimal solution of the \urap\ instance on $G$ with $f_1,f_2$
  being vulnerable has cardinality of $\frac{n}{2} +\frac{L}{2} + 2$.
  This shows that \snpp\ is reducible to \urap\ with two vulnerable
  edges, in polynoimal time. 
\end{proof} 

\bigskip
The remainder of this section addresses the \np-hardness of 
 \snpp\, stated in the next lemma.
\begin{lemma}\label{thm:SNP_np_hard}
  \snpp\ is \np-hard.
\end{lemma}

The proof of Lemma~\ref{thm:SNP_np_hard} resorts to a reduction from a special variant of the
following well-known \np-complete decision problem.
\begin{problem*}[{Path with Forbidden Pairs Problem (\pfp),~\cite[GT54]{garey_johnson_79}}]
  \text{}
  \begin{itemize}
  \item \underline{Input:} Tuple $(D,s,t,\{\{u_i,v_i\} \mid i \in
    \oneto{k} \})$, where $D = (V,A)$ is a directed graph, $s,t$ are two terminals in $V$, 
    and $\{u_i, v_i\}$ are $k$ pairs of nodes in $V$ (with all $2k$
    nodes being  distinct).
  \item \underline{Question:} 
    Does $D$ contain an $s$-$t$-path $P$ with the property
    that, for every $i\in [k]$, at most one node from $\{u_i, v_i\}$ is
    contained in $P$?
  \end{itemize}
\end{problem*}
To show \np-hardness of \snpp\ the following restricted version of \pfp\ is considered.
\begin{problem*}[Restricted Path with Forbidden Pairs Problem (\rpfp)]
\text{}
  \label{prob:RPFP}
  \begin{itemize}
  \item \underline{Input:}
    Tuple $(H,s,t,\{\{u_i,v_i\} \mid i \in
    \oneto{k} \})$, where   
    \begin{enumerate}
    \item[{\small (P1)}] $H = (U\dcup W, E_H)$ is undirected,  bipartite and balanced.
    \item[{\small (P2)}] $s\in U$ and $t\in W$.
    \item[{\small (P3)}] $k$ is even.
    \item[{\small (P4)}] $\{u_i, v_i\} \subseteq U$ or $\{u_i, v_i\} \subseteq W$ for all $i\in [k]$.
    \item[{\small (P5)}] $\big|\big\{i\in\{1,\hdots,k\}\mid u_i,v_i\in U\big\}\big|
      = \big|\big\{i\in\{1,\hdots,k\}\mid u_i,v_i\in W\big\}\big| = \frac{k}{2}$.
    \end{enumerate}
  \item\underline{Question:}
    Does $H$ contain an $s$-$t$-path $P$ with the property
    that, for every $i\in [k]$, at most one node from $\{u_i, v_i\}$ is
    contained in $P$?
\end{itemize}
\end{problem*}

It turns out that the special version \rpfp\ remains \np-complete.
\begin{lemma}
  \label{lem:special-SPFP-is-np-hard}
  \rpfp\ is \np-complete.
\end{lemma}
\begin{proof}
To prove this statement, the proof of the \np-completeness result for  
\pfp\  as given by Gabow et al.~\cite[Lemma 2]{gabow_et_al_76} can be adapted
because of which it is only sketched here.\footnote{In \cite{gabow_et_al_76}, \pfp\ is called
  \emph{Impossible Pairs Constrained Path Problem}.}\\
The main idea is to use a reduction from the \np-complete
\emph{3-Satisfyability Problem} \threesat. This is achieved by
modeling pairs of a literal and its negation as forbidden pairs. 
For each clause appearing in the Boolean formula that specifies the \threesat\ instance,
a layer of nodes is further introduced.
Each two neighboring layers are then connected via a complete
bipartite graph. A path avoiding the forbidden pairs yields a truth assignment and vice versa.\\
The main difference to the proof given in~\cite[Lemma
2]{gabow_et_al_76} is that, for each clause in the Boolean formula,
two identical layers of nodes needs to be introduced in order to satisfy the
properties {\small (P4)} and {\small (P5)} of \rpfp. 
\end{proof}

The section is concluded by presenting the proof on the \np-hardness
of \snpp.
\begin{proof}[Proof of Lemma~\ref{thm:SNP_np_hard}]  
 As already mentioned, the proof is based on a polynomial reduction
 from \rpfp\ to \snpp\ that is shown next.\\ 
 Let $H = (U\dcup W, E_H)$ be a balanced bipartite graph, $s\in U, t\in W$,
 and let $\{(u_1,v_1), \dots, (u_k,v_k)\}$ be a collection of
 forbidden pairs, all together comprising an instance $\inst$ of \rpfp. 
 To obtain a corresponding \snpp\ instance $\inst^\prime$
 the following five steps are performed on $H$, and illustrated in Figure~\ref{fig_transformation_rpfp_to_snnp}.
 \begin{figure}[htb] 
   \resizebox{0.99\textwidth}{!}{\input{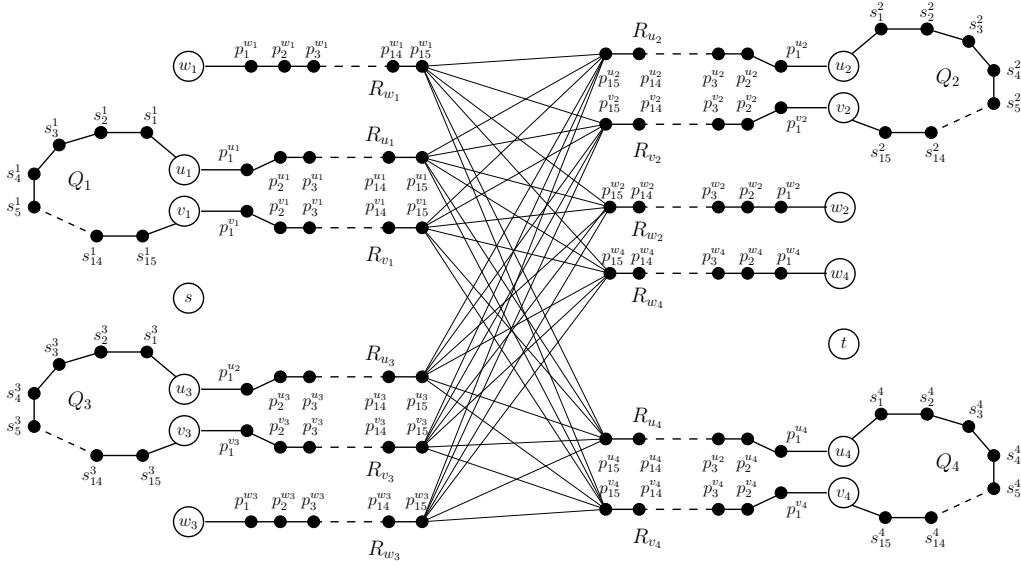}}
   \caption{%
     For a given \rpfp\ instance $(H,s,t,\{\{u_i,v_i\}\mid i\in[4]\})$
     on a balanced bipartite graph $H:=(U\dcup W,E_H)$ with
     $U:=\{s,u_1,v_1,u_3,v_3,w_1,w_3\}$ and
     $W:=\{t,u_2,v_2,u_4,v_4,w_2,w_4\}$, the figure shows the graph
     $H^\prime$ derived from $H$ by applying the transformation steps
     {\small (S1)--(S5)}. Larger, unfilled nodes represent the original nodes from
     $H$, while the new nodes introduced through the transformation
     steps appear as small black nodes.
     To improve the readability of the figure, the edge set $E_H$ from the original graph $H$
     is omitted. Note that $L=15$, and recap that $s^\prime=s$ and
     $t^\prime=t$ hold.
   } 
   \label{fig_transformation_rpfp_to_snnp}
 \end{figure}
\begin{enumerate}
\item[{\small (S1)}]
  Set $L := |U\dcup W| + 1$ (that is an odd number as $|U\dcup W|$ is even). 
\item[{\small (S2)}]
  For every $i \in \oneto{k}$, introduce new nodes $s^i_{j}$, $j=1,\hdots,L$,
  as well as the new path $Q_i:=(u_i, s^i_{1},s^i_{2}, \dots,
  s^i_{L}, v_i)$ connecting the two nodes of the corresponding forbidden pair
  $\{u_i,v_i\}$ through the new nodes $s^i_{j}$, $j=1,\hdots,L$, 
  and add both the new nodes and the edges of $Q_i$ to $H$.
\item[{\small (S3)}]
  For every node $w \in (U\dcup W) \setminus \{s,t\}$, introduce a new
  path $R_w:=(w,p^w_{1}, p^w_{2},\dots, p^w_{L})$, and add the new
  nodes $p_{q}^w$, $q=1,\hdots,L$ and the edges of $R_w$ to $H$. 
\item[{\small (S4)}] 
  Define $T := \{p_L^w \mid  w\in (U\dcup w)\setminus \{s,t\}\}$, and
  add to $H$ the edge set
  $$
  E_T := \big\{\, \{ p_L^{w_1}, p_L^{w_2}\}\,\big\} \mid  w_1\in U
  \setminus \{s\}, w_2 \in W \setminus \{t\}\}.
  $$
  This yields a complete, balanced, bipartite subgraph
  $H_T:=(T,E_T)$. 
\item[{\small (S5)}] Set $s^\prime = s$ and $t^\prime = t$.
\end{enumerate} 
Let $H^\prime:=(U^\prime\dcup W^\prime,E_{H^\prime})$ be the graph after performing
steps {\small (S1)} to {\small (S5)}.
The bipartiteness of $H^\prime$ follows from property {\small (P4)} of \rpfp,
i.e.\ from the fact that, for each forbidden pair
$(u_i, v_i)$, both $u_i$ and $v_i$ belong to the same node partition
of $H$.
By property {\small (P5)} of \rpfp, the number of forbidden pairs contained in each bipartition is
the same. Thus, $H^\prime$ is also balanced.
Moreover, $H^\prime$ contains $H$ as a subgraph, and it may be assumed
that $U\subseteq U^\prime$ and $W\subseteq W^\prime$,
i.e.\ $s^\prime \in U^\prime$ and $t^\prime \in W^\prime$.
Therefore, $\inst^\prime:=(H^\prime,s^\prime,t^\prime)$ is an
instance of \snpp\ constructed from $\inst$ in polynomial time.\\
Note further that each sub-path $(s^i_{1}, s^i_{2}, \dots, s^i_{L})$
introduced in step {\small (S2)} contains an odd number of
nodes. Thus, the nodes sets of the form $\{s^i_{1}, s^i_{2}, \dots,
s^i_{L}\}$ cannot be perfectly matched among themselves in
$H^\prime$. The same is true for the sub-paths $(p^w_{1},
p^w_{2},\dots, p^w_{L})$ introduced in step {\small (S3)}.  
Observe also that, for each $w \in T$, the nodes $p_L^w$ and $w$ are in different parts of the bipartition,
which is due to the fact that $L$ is odd.

Now consider the following claim.
\begin{itemize}
\item[(C)]\snpp\ instance $\inst^\prime$ contains an $s$-$t$-path of length $\ell < L$ if and only if the \rpfp\ 
instance $\inst$ is feasible.
\end{itemize}

Together with Lemma~\ref{lem:special-SPFP-is-np-hard}, the correctness
of claim (C) implies that finding a shortest nice path is \np-hard, thus completing
the proof of this lemma.

\emph{``only if'' part of the claim (C)}.\quad
Let $P$ be a nice $s$-$t$-path of length $\ell < L$ in $H^\prime$. 
Since $P$ has at most $L-1$ edges, it can not contain an edge from one of
the paths $Q_i$, $i\in [k]$, and $R_w$, $w\in (U\dcup
W)\setminus\{s,t\}$, as otherwise all 
edges from $Q_i$ or $R_w$ must be contained in $P$. This, however, implies that 
the length of $P$ is greater than $L$ as all paths $Q_i$ and $R_w$
consist of $L+1$ edges (see steps {\small (S2)} and {\small (S3)}),
leading to a contradiction.\\  
Furthermore, since $P$ contains no edge from a path $R_w$, it is also
disjoint from the subgraph $H_T$ introduced in step {\small (S4)}.
It follows that $P$ is completely contained the subgraph $H$ of
$H^\prime$, i.e.\ $P$ is already a path in $H$. \\
Next, it is shown that $P$ is feasible to $\inst$. This is a achieved
by proving that, for every $i\in [k]$, at least one of the nodes $u_i, v_i$ 
is not incident to $P$. For this, let $\{u_i, v_i\}$ be any forbidden
pair and let $V[P]$ be the node set of $P$. Recall that any node from
$V[P]$ is a node in $H$. 
Since $P$ is a nice path in $H^\prime$, there exists a matching $M$ in
$H^\prime$ covering all nodes from $H^\prime$ that are not contained
in $V[P]$.
As the inner nodes $s^i_{1}, s^i_{2}, \dots, s^i_{L}$ of any path
$Q_i$ does not belong to $P$, it follows that they are covered by $M$.
Recap moreover that $L$ is odd. This, in particulary, means that either $\{s^i_{1},u_i\}\in M$ or $\{s^i_{L},v_i\}\in M$ holds.
Thus, either $u_i$ or $v_i$ is incident to $P$.

\emph{``if'' part of the claim (C)}.\quad
Let $P$ be a feasible solution to the \rpfp\ instance $\inst$, i.e.\  
$P$ is a path in $H$.
To show that $P$ is a nice path in $H^\prime$ it must be proved that
there is a matching $M$ in $H^\prime$ that covers all nodes except
those from $V[P]$.
For this we assume w.l.o.g. that at most the nodes $u_i$ are part of the path $P$ and define the following index sets.
\begin{align*}
  J := \{i\in [k]\,\mid\, u_i\in V[P]\},\;\;
  K := [k]\setminus J. 
\end{align*}
Note, that the sets $J$ and $K$ form a partition of $[k]$.\\
The matching $M$ is constructed as follows. Set $M:=\emptyset$, and
consider first the nodes on the paths $Q_i$, $i\in [k]$.
To cover all nodes of $Q_i$ not belonging to $V[P]$, a suitable
set of alternating edges (i.e.\ every second edge) from the path $Q_i$ is chosen. 
Concretely, if $i\in J$, add to $M$ the set of alternating edges
covering $s^i_{1}, \dots, s^i_{L}, v_i$.   
Otherwise, i.e.\ if $i\in K$, add to $M$ the set of alternating
edges that match $u_i, s^i_{1}, \dots, s^i_{L}$. 
Observe that $M$ already covers either $u_i$ or $v_i$, for each
forbidden pair $\{u_i,v_i\}$. By properties~{\small (P4)} and~{\small (P5)}
of \rpfp\, it further follows that the number of nodes in $U$ that are
covered by $M$ and that belong to a forbidden pair is identical to
the number of nodes in $W$ covered by $M$ and belonging to a forbidden
pair.  
\\
Secondly, consider all nodes in subgraph $H$ that do not belong to $V[P]$
and that are not covered by $M$, so far. 
Let $\tilde U \cup \tilde W$ denote these nodes where $\tilde U\subseteq
U$ and $\tilde W\subseteq W$.
Then, $|\tilde U| = |\tilde W|$ holds, which follows from the fact
that  
$$
|\tilde U| =  \frac{n}{2} - q - \frac{k}{2} = |\tilde W|, 
$$
where
\begin{itemize}
 \item $\frac{n}{2} = |U| = |W|$ is the number of nodes on each side of $H$,
 \item $q = \frac{1}{2} |V[P]|$ (Note that $|V[P]|$ is even as $s\in U$
   and $t\in W$ and $P$ is an $s$-$t$-path in $H$. More precisely,
   even $|V[P]\cap U|=|V[P]\cap W|$ holds.), 
 \item and $\frac{k}{2}$ the number of nodes in $W$ (as well as in $U$)
   that are covered by $M$ and that belong to a forbidden pair.
\end{itemize}
To extend $M$ to a matching in $H^\prime$ also covering $\tilde U \cup
\tilde W$, the following edges are added to $M$.
For each $w\in \tilde U \cup \tilde W$, choose the set of
alternating edges from path $R_w$ that cover $w,p^w_{1}, p^w_{2},\dots, p^w_{L}$.\\
Now observe that the only nodes in $H^\prime$ not belonging to $P$ and being
still unmatched by $M$ are all but the first nodes  $\{p^w_{1}, p^w_{2},\dots,
p^w_{L}\}$ of a path $R_w$ that is associated with a node $w\in (U\cup W) \setminus
(\tilde{U}\cup\tilde{W}\cup\{s,t\})$, i.e.\ with a node $w$ in $H$
that is either contained in $V[P]$ or $w$ is a node that is
part of a forbidden pair and that is covered by an edge added to $M$
in the first step.
These nodes can be covered by first adding, for each $w\in (U\cup W) \setminus
(\tilde{U}\cup\tilde{W}\cup\{s,t\})$,  the edges $\{p^w_{1}, p^w_{2}\}$,
$\{p^w_3,p^w_4\}, \ldots$, $\{p^w_{L-2},p^w_{L-1}\}$ from $R_w$ to
$M$.\\
This still leaves all end nodes $p^w_{L}$ of the paths $R_w$ with a node $w\in (U\cup W) \setminus
(\tilde{U}\cup\tilde{W}\cup\{s,t\})$ to be unmatched by $M$. Let
 $$ T^\prime:=\{ p^w_L \mid w\in (U\cup W) \setminus
(\tilde{U}\cup\tilde{W}\cup\{s,t\})\} \subseteq T
$$ 
be the set of all these end nodes. 
It remains to extend $M$ to a matching also covering $T^\prime$.
Recap that $T^\prime$ is,  as a subset of $T$, a part of the
balanced bipartite subgraph $H_T=(T,E_T)$ constructed in step {\small (S4)}. 
Furthermore, it holds that $|T^\prime \cap U| =
\tfrac{1}{2}|T^\prime|= |T^\prime\cap W|$.
Therefore, $H_T$ contains a matching only covering nodes
from $T^\prime$. 
After adding one such matching on $T^\prime$ to $M$, the set $M$
becomes a matching in $H^\prime$ that covers all nodes of $H^\prime$
not contained in $P$. This shows that $P$ is a nice $s$-$t$-path in
$H^\prime$. As $P$ is completely contained in the subgraph $H$, its
node set $V[P]$ can only consists of at most $|U\dcup W|< |U\dcup W|+1=:L$ nodes,
proving that the length $\ell$ of $P$ is strictly less than $L$.  
\end{proof}


\section{Conclusion and Future Work}

This paper studies a novel practically relevant robust variant of the assignment problem (\rap).
Tight connections between \rap\ and classical notions in matching
theory, including matching-covered graphs and ear decompositions, have
been highlighted and used to obtain asymptotically tight approximation results for \rap. 
The approximation algorithm presented for the general variant of
\rap\ combines classical results for matching-covered graphs with \lp\ randomized rounding techniques.

Some ongoing and future work includes the following lines of research. 
Study a version of \rap\ with node failures, or with a combination of node and edge failures. 
This problem has many potential applications beyond the ones listed
here. Study the variant of \rap\ where each scenario consists of at
most $k$ edges, for some input parameter $k > 1$. This paper treats the case $k=1$. 
Besides, it is interesting to study the complexity of \rap\ in general graphs.

\bigskip
\paragraph*{Acknowledgement. }
The work of the second author and the third author is part of the Research Training Group \emph{``Discrete Optimization of Technical Systems under Uncertainty''} (RTG 1855).
Financial support through RTG 1855 of the second author by the German Research Foundation (DFG) is gratefully appreciated.
 
{

  \bibliographystyle{plainnat}
}  
\end{document}